\UseRawInputEncoding
\documentclass[12pt, draftclsnofoot, journal, letterpaper, onecolumn]{IEEEtran}

\makeatletter
\def\ps@headings{%
\def\@oddhead{\mbox{}\scriptsize\rightmark \hfil \thepage}%
\def\@evenhead{\scriptsize\thepage \hfil \leftmark\mbox{}}%
\def\@oddfoot{}%
\def\@evenfoot{}}
\makeatother 

\IEEEoverridecommandlockouts
\usepackage{bbm}
\usepackage{amsfonts}
\usepackage[dvips]{graphicx}
\usepackage{times}
\usepackage{cite}
\usepackage{amsmath}
\usepackage{array}
\usepackage{amssymb}

\usepackage{stfloats}
\usepackage{slashbox}
\usepackage{graphicx}
\usepackage{subfigure}
\usepackage{footnote}
\usepackage{booktabs}
\usepackage{array}
\usepackage{algorithm}
\usepackage{subeqnarray}
\usepackage{cases}
\usepackage{threeparttable}
\usepackage{color}
\usepackage{epstopdf}
\usepackage{algpseudocode}
\usepackage{bm}
\usepackage{multirow}
\usepackage{adjustbox}
\newtheorem{property}{Property}
\newtheorem{theorem}{Theorem}
\newtheorem{remark}{Remark}

\newtheorem{lemma}{Lemma}

\newtheorem{corollary}{Corollary}
\allowdisplaybreaks[4]

\begin{document}

\title{Gradient Statistics Aware Power Control for Over-the-Air Federated Learning}

\author{Naifu Zhang and Meixia Tao
\thanks{N. Zhang and M. Tao are with the Department of Electronic Engineering, Shanghai Jiao Tong University, Shanghai,
200240, P. R. China. (email: arthaslery@sjtu.edu.cn; mxtao@sjtu.edu.cn.). Part of this work is presented at IEEE ICC 2020 workshop \cite{zhang2020gradient}}}

\maketitle
\vspace{-1.5cm}
\begin{abstract}
Federated learning (FL) is a promising technique that enables many edge devices to train a machine learning model collaboratively in wireless networks.
By exploiting the superposition nature of wireless waveforms, over-the-air computation (AirComp) can accelerate model aggregation and hence facilitate communication-efficient FL.
Due to channel fading, power control is crucial in AirComp.
Prior works assume that the signals to be aggregated from each device, i.e., local gradients have identical statistics.
In FL, however, gradient statistics vary over both training iterations and feature dimensions, and are unknown in advance.
This paper studies the power control problem for over-the-air FL by taking gradient statistics into account.
The goal is to minimize the aggregation error by optimizing the transmit power at each device subject to peak power constraints. We obtain the optimal policy in closed form when gradient statistics are given.
Notably, we show that the optimal transmit power is continuous and monotonically decreases with the squared multivariate coefficient of variation (SMCV) of gradient vectors.
We then propose a method to estimate gradient statistics with negligible communication cost.
Experimental results demonstrate that the proposed gradient-statistics-aware power control achieves higher test accuracy than the existing schemes for a wide range of scenarios.
\end{abstract}


\section{Introduction}
The proliferation of mobile devices such as smartphones, tablets, and wearable devices has revolutionized people's daily lives.
Due to the growing computation and sensing capabilities of these devices, a wealth of data has been generated each day.
This has promoted a wide range of artificial intelligence (AI) applications such as image recognition and natural language processing.
Traditional machine learning procedures, including both training and inference, are carried out through cloud computing on a centralized data center that has full access to the entire data set.
Wireless edge devices are thus required to upload their collected raw data to the center, which can be very costly in terms of energy and bandwidth consumption, and unfavorable due to response delay and privacy concerns.
It is thus increasingly desired to let edge devices engage in the learning process by keeping the raw data locally and performing training/inference either collaboratively or individually.
This emerging technology is known as Edge Machine Learning \cite{park2019wireless} or Edge Intelligence \cite{zhou2019edge}.

Federated learning (FL)\cite{mcmahan2016communication,konen2016federated,bonawitz2019towards,yang2019federated} is a new edge learning framework that enables many edge devices to collaboratively train a machine learning model without exchanging datasets under the coordination of an edge server in wireless networks.
Compared with traditional learning at a centralized data center, FL offers several distinct advantages, such as preserving privacy, reducing network congestion, and leveraging distributed on-device computation.
In FL, each edge device downloads a shared model from the edge server, computes an update to the current model by learning from its local dataset, then sends this update to the edge server.
Therein, the updates are averaged to improve the shared model.
FL has recently attracted significant attention from both academia and industry, such as \cite{9026922,9003425,9084352,tao2020zte}.

The main bottleneck in FL is the communication cost since a large number of participating edge devices send their updates to the edge server at each round of the model training.
Existing methods to obtain communication-efficient FL can be mainly divided into three categories: model parameter compression \cite{alistarh2018qsgd,seide20141}, gradient sparsification \cite{aji2017sparse,tsuzuku2018variance}, and infrequent local update \cite{mcmahan2016communication,wang2019adaptive}.
Nevertheless, the communication cost of FL is still proportional to the number of edge devices, and thus inefficient in large-scale environment.
Recently, a fast model aggregation approach is proposed for synchronous FL by applying the over-the-air computation (AirComp) principle \cite{nazer2007computation}, such as in \cite{yang2018federated,zhu2018low,amiri2019machine,amiri2019federated}.
This is accomplished by exploiting the waveform superposition nature of the wireless medium to compute the desired function of the distributed local gradients (i.e., the weighted average function) by concurrent transmission.
Such AirComp-based FL, referred to as \emph{over-the-air FL} in this work, can dramatically save the uplink communication bandwidth.

Due to the channel fading, device selection and power control are crucial to achieve a reliable and high-performance over-the-air FL.
In \cite{cao2019optimal}, the authors jointly optimize the transmit power at edge devices and the receive scaling factor (known as denoising factor) at the edge server for minimizing the mean square error (MSE) of the aggregated signal.
It is shown that the optimal transmit power in static channels exhibits a threshold-based structure.
Namely, each device applies channel-inversion power control if its quality indicator exceeds the optimal threshold, and applies full power transmission otherwise.
The work \cite{cao2019optimal}, however, is purely for AirComp-based signal aggregation (not in the context of learning), where the signal on each device is assumed to be independent and identically distributed (IID) with zero mean and unit variance.
For AirComp-based gradient aggregation in FL, the work \cite{zhu2018low} introduces a truncation-based approach for excluding the edge devices with deep fading channels to strike a good balance between learning performance and aggregation error.
The work \cite{yang2018federated} proposes a joint device selection and receiver beamforming design method to find the maximum number of selected devices with MSE requirements to improve the learning performance.
As in \cite{cao2019optimal}, it is assumed in both \cite{yang2018federated} and \cite{zhu2018low} that the signal (i.e., the local gradient) to be aggregated from each device is IID, and normalized with zero mean and unit variance.
By exploiting the sparsity pattern in gradient vectors, the work \cite{amiri2019federated} projects the gradient estimate in each device into a low-dimensional vector and transmits only the important gradient entries while accumulating the error from previous iterations.
Therein, a channel-inversion like power control scheme, similar to those in \cite{cao2019optimal,zhu2018low,yang2018federated} is designed so that the gradient vectors sent from the selected devices are aligned at the edge server.

Note that all the exiting works on power control for over-the-air FL have overlooked the following statistical characteristics of gradients:
\emph{the gradient distribution over training iterations is not necessarily identical; and even in the same iteration, the distribution of each entry of the gradient vector can be non-identical.}
A general observation is that the gradient distribution changes over iterations and is different in each feature dimension.
In addition, if the gradient distribution is unknown for each device, normalizing the gradient to a distribution with zero mean and unit variance is infeasible.
As such, due to the neglect of the above characteristics of gradient distribution, the existing power control methods for over-the-air FL may not perform efficiently in practice.

Motivated by the above issue, in this paper, we study the optimal power control problem for over-the-air FL in fading channels by taking gradient statistics into account.
Our goal is to minimize the MSE of the aggregated model at each iteration, and hence improve the accuracy of FL, by jointly optimizing the transmit power at each device and the denoising factor at the edge server given the first-order and second-order statistics of gradients.
The main contributions of this work are outlined below:
\begin{itemize}
\item\emph{Optimal power control with known gradient statistics:}
We first derive the MSE expression of gradient aggregation at each iteration of the model training when the first-order and second-order statistics of the gradient vectors are known.
We then formulate a joint optimization problem of transmit power at edge devices and denoising factor at the edge server for MSE minimization subject to individual peak power constraints at edge devices.
By decomposing this non-convex problem into subproblems defined on different subregions, we obtain the optimal power control strategy in closed form.
Unlike the existing threshold-based power control for normalized signal in \cite{cao2019optimal}, our optimal transmit power depends not only on the channel quality and noise power, but also heavily on the gradient statistics.
In particular, we find that the relative dispersion of the gradient values, i.e., the squared multivariate coefficient of variation (SMCV) of the gradient vectors, plays a key role in the optimal transmit power control.
Specifically, we prove that the optimal transmit power of each device is a continuous and monotonically decreasing function of the gradient SMCV.

\item\emph{Optimal power control in special cases:}
In the special case where the gradient SMCV approaches infinity, which could happen when the model training converges and/or the dataset is highly non-IID, we show that there is an optimal threshold for the device aggregation capability, below which the device transmits with full power and above which it transmits at the power to equalize the weight of its gradient for aggregation to one.
In the other special case where the gradient SMCV approaches zero, which could happen when the model training just begins, we show that the optimal power control is to let all the devices transmit with their peak powers.

\item\emph{Adaptive power control with unknown gradient statistics:}
We propose an adaptive power control algorithm that estimates the gradient statistics based on the historical aggregated gradients and then dynamically adjusts the power values in each iteration based on the estimation results.
The communication cost consumed by estimating the gradient statistics is negligible compared to the transmission of the entire gradient vector.
\end{itemize}

To evaluate the efficiency of the proposed power control scheme, we implement the FL in PyTorch for AI applications of three datasets: MNIST, CIFAR-10 and SVHN.
Experimental results demonstrate that the over-the-air FL with the proposed adaptive power control obtains higher model accuracy than the considered existing power control methods (full power transmission and threshold-based power control for normalized signal \cite{cao2019optimal}).
In particular, while the full power transmission performs poorly in high signal-noise ratio (SNR) region and non-IID data distribution and the threshold-based power control for normalized signal \cite{cao2019optimal} performs poorly in low SNR region and IID data distribution, the proposed power control can perform very well over a wide range of scenarios by exploiting the gradient statistics.

The rest of this paper is organized as follows.
The over-the-air FL system is modeled in Section II.
Section III describes the optimal power control strategy with known gradient statistics.
In Section IV, we introduce an adaptive power control scheme when the gradient statistics are unknown in advance.
Section V provides experimental results.
Finally we conclude the paper in Section VI.

\section{System Model}
\subsection{Federated Learning Over Wireless Networks}
\begin{figure}[t]
\begin{centering}
\vspace{-0.2cm}
\includegraphics[scale=.50]{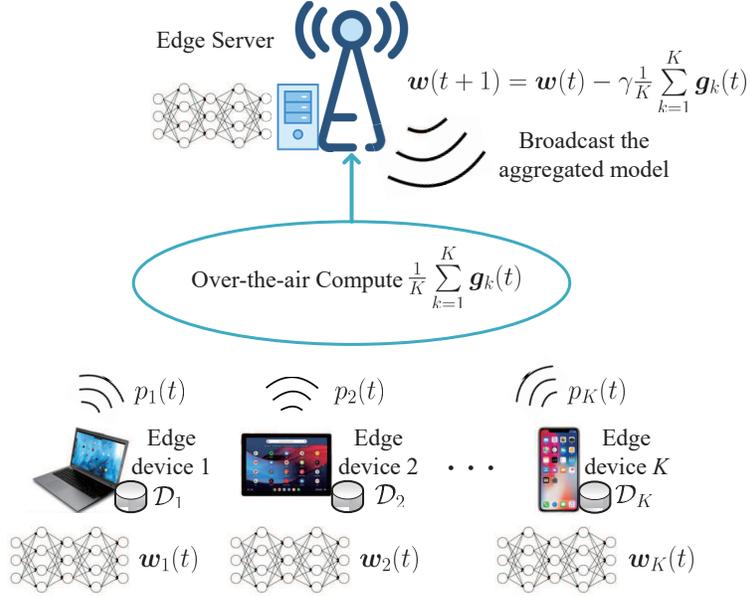}
\vspace{-0.1cm}
 \caption{\small{Illustration of over-the-air federated learning.}}\label{fig:system}
\end{centering}
\vspace{-0.3cm}
\end{figure}
We consider a wireless FL framework as illustrated in Fig.~\ref{fig:system}, where a shared AI model (e.g., a classifier) is trained collaboratively across $K$ edge devices via the coordination of an edge server.
Let $\mathcal{K}=\{1,...,K\}$ denote the set of edge devices.
Each device $k\in\mathcal{K}$ collects a fraction of labelled training data via interaction with its own users, constituting a local dataset, denoted as $\mathcal{S}_k$, which is unknown to the edge server.
Let $\bm{w}\in\mathbb{R}^D$ denote the $D$-dimensional model parameter to be learned.
The loss function measuring the model error is defined as
\begin{equation}
L(\bm{w})=\sum_{k\in\mathcal{K}}{\frac{|\mathcal{S}_k|}{|\mathcal{S}|}L_k(\bm{w})},
\label{equ:loss_function}
\end{equation}
where $L_k(\bm{w})=\frac{1}{|\mathcal{S}_k|}\sum_{i\in\mathcal{S}_k}{l_i(\bm{w})}$ is the loss function of device $k$ quantifying the prediction error of the model $\bm{w}$ on the local dataset collected at the $k$-th device, with $l_i(\bm{w})$ being the sample-wise loss function, and $\mathcal{S}=\bigcup_{k\in\mathcal{K}}{\mathcal{S}_k}$ is the union of all datasets.
The minimization of $L(\bm{w})$ is typically carried out through the stochastic gradient descent (SGD) algorithm, where device $k$'s local dataset $\mathcal{S}_k$ is split into mini-batches of size $B$ and at each iteration $t=1,2,...$, we draw one mini-batch $\mathcal{B}_k(t)$ randomly and update the model parameter as
\begin{equation}
\bm{w}(t+1)=\bm{w}(t)-\gamma\frac{1}{K}\sum_{k\in\mathcal{K}}{\nabla L_{k,t}^{SGD}(\bm{w}(t))},
\end{equation}
with $\gamma$ being the learning rate and $L_{k,t}^{SGD}(\bm{w})=\frac{1}{B}\sum_{i\in\mathcal{B}_k(t)}{l_i(\bm{w})}$.
Note that the mean of the gradient $\nabla L_{k,t}^{SGD}(\bm{w}(t))$ in SGD is equal to the gradient $\nabla L_k(\bm{w(t)})$ in gradient descent (GD).

\subsection{Over-the-Air Computation for Gradient Aggregation}
Let $\bm{g}_k(t)\triangleq\nabla L_{k,t}^{SGD}(\bm{w}(t))\in\mathbb{R}^D$ denote the gradient vector computed on device $k$ at iteration $t$.
The following are key assumptions on the distribution of each entry $g_{k,d}(t),d\in\{1,...,D\}$ of $\bm{g}_k(t)$:
\begin{itemize}
\item The gradient elements $\{g_{k,d}(t)\}, \forall k\in \mathcal{K}$, are independent and identically distributed over devices $k$'s.
Note that the local datasets are typically non-IID across devices in federated learning. However, due to privacy concerns, their actual distributions are usually unknown to the edge server, whether being identical or non-identical. As such and, by default, the distributions of the local gradients $\{g_{k,d}(t)\}$ trained from these local datasets are treated equally across devices from the edge server's perspective.

\item The gradient elements $\{g_{k,d}(t)\}, \forall t\in\mathbb{N}$, are non-identically distributed over iterations $t$'s.
In other words, the gradient distribution is non-stationary over time.
The non-stationary distribution is valid since the gradient values in general change rapidly at the beginning, then gradually approach to zero as the training goes on.

\item The gradient elements $\{g_{k,d}(t)\}, \forall d\in\{1,2,...,D\}$, are independent but non-identically distributed over gradient vector dimension $d$'s.
This assumption is valid as long as the features in a data sample are independent but non-identically distributed, which is typically the case.
\end{itemize}
The gradient of interest at the edge server at each iteration $t$ is given by
\begin{equation}
\label{equ:average_gradient}
\bm{g}(t)=\frac{1}{K}\sum_{k\in\mathcal{K}}{\bm{g}_k(t)}.
\end{equation}

To obtain (\ref{equ:average_gradient}), all the devices transmit their gradient vectors $\bm{g}_k(t)$ concurrently in an analog manner, following the AirComp principle as shown in Fig.~\ref{fig:system}.
We consider block fading channels, where the channel coefficients remain unchanged within the duration of each iteration in FL but may change independently from one iteration to another.
We define the duration of one iteration as one time block, indexed by $t\in\mathbb{N}$.
For simplicity, each edge device and the edge server are equipped with a single antenna.
Let $h_k(t)$ denote the complex-valued channel coefficient from device $k$ to the edge server at the $t$-th time block.
It is assumed to be generated from a stationary and ergodic process.
Each transmission block takes a duration of $D$ slots, one slot for one entry in the $D$-dimensional gradient vector.
Each gradient vector $\bm{g}_k(t)$ is multiplied with a pre-processing factor, denoted as $b_k(t)$.
The received signal vector at the edge server is given by
\begin{equation}
\bm{y}(t)=\sum_{k\in\mathcal{K}}{h_k(t)b_k(t)\bm{g}_k(t)}+\bm{n}(t),
\label{equ:received signal}
\end{equation}
where $\bm{n}(t)$ denotes the additive white Gaussian noise (AWGN) vector at the edge server with each element having zero mean and variance of $\sigma_n^2$.
Let $p_k(t)\geq0$ denote the transmit power of device $k\in\mathcal{K}$ at time block $t$.
To compensate the channel phase offset and scale the actual signal amplitude at each device, we design the pre-processing factor as $b_k(t)=\sqrt{\frac{p_k(t)}{\alpha_k(t)}}e^{-j\theta_k(t)}$, where $\theta_k(t)$ is the phase of $h_k(t)$ and $\alpha_k(t)\triangleq\mathbb{E}[\|\bm{g}_k(t)\|^2]$ is the mean squared norm of the gradient of device $k$.
Here, we have assumed that each device $k$ can estimate perfectly its own channel phase $\theta_k(t)$.
The mean squared norm of gradient, $\alpha_k(t)$ can be estimated with negligible communication costs as we shall present in Section IV.
To design the optimal power control policy $p_k(t)$, for $k\in\mathcal{K}$, we also assume that the edge server knows the channel amplitude $|h_k|$ of all devices.
By such design of $b_k(t)$, we can rewrite (\ref{equ:received signal}) as
\begin{equation}
\bm{y}(t)=\sum_{k\in\mathcal{K}}{\sqrt{\frac{p_k(t)}{\alpha_k(t)}}|h_k(t)|\bm{g}_k(t)}+\bm{n}(t).
\end{equation}
Each device $k\in\mathcal{K}$ has a peak power budget $P_k$, i.e.,
\begin{equation}
p_k(t)\leq P_k,\quad\forall{k}\in\mathcal{K}, \forall t\in\mathbb{N}.
\label{equ:power_constrain}
\end{equation}
Upon receiving $\bm{y}(t)$, the edge server applies a denoising factor, denoted by $\eta(t)$, to recover the gradient of interest as
\begin{equation}
\bm{\hat{g}}(t)=\frac{\bm{y}(t)}{K\sqrt{\eta(t)}},
\label{recover_gradient}
\end{equation}
where the factor $1/K$ is employed for the averaging purpose.

\subsection{Performance Measure}
We are interested in minimizing the distortion of the recovered gradient $\bm{\hat{g}}(t)$ with respect to (w.r.t.) the ground true gradient $\bm{g}(t)$ at each iteration $t$.
The distortion is measured by the instantaneous MSE defined as
\begin{align}
\text{MSE}(t)&\triangleq~\mathbb{E}[\|\bm{\hat{g}}(t)-\bm{g}(t)\|^2]\nonumber\\
&=~\frac{1}{K^2}\mathbb{E}\left[\left\| \frac{\bm{y}(t)}{\sqrt{\eta(t)}}-\sum_{k\in\mathcal{K}}{\bm{g}_k(t)} \right\|^2\right]\nonumber\\
&=~\frac{1}{K^2}\Bigg[\sum_{d=1}^D{\sigma_d^2(t)}\sum_{k\in\mathcal{K}}{\left(\frac{\sqrt{p_k(t)}|h_k(t)|}{\sqrt{\eta(t)\sum_{d=1}^D{(\sigma_d^2(t)+m_d^2(t))}}}-1\right)^2}\nonumber\\
&~~~~~~~~~~+\sum_{d=1}^D{m_d^2(t)}{\left(\frac{\sum_{k\in\mathcal{K}}\sqrt{p_k(t)}|h_k(t)|}{\sqrt{\eta(t)\sum_{d=1}^D{(\sigma_d^2(t)+m_d^2(t))}}}-K\right)^2}+\frac{D\sigma_n^2}{\eta(t)} \Bigg],
\label{equ:MSE}
\end{align}
where the expectation is over the distribution of the transmitted gradients $\bm{g}_k(t)$ and the received noise $\bm{n}(t)$, $m_d(t)$ and $\sigma_d^2(t)$ denote the mean (first-order statistics) and variance (second-order statistics) of the $d$-th entry of gradient $\bm{g}(t)$ at iteration $t$, respectively.
Note that in deriving (\ref{equ:MSE}), we have used the previous assumption that the gradient elements $\{g_{k,d}(t)\}, \forall k\in \mathcal{K}$, are IID across devices, i.e, $\alpha_k(t)=\mathbb{E}[\|\bm{g}_k(t)\|^2]=\mathbb{E}[\|\bm{g}(t)\|^2]=\sum_{d=1}^D{\left(\sigma_d^2(t)+m_d^2(t)\right)}$.

Observing (\ref{equ:MSE}) closely, we find that the MSE consists of three components, representing the individual misalignment error (the first term), the composite misalignment error (the second term), and the noise-induced error (the third term), respectively.
The individual misalignment error is weighted by the gradient variance $\sum_{d=1}^D{\sigma_d^2(t)}$ while the composite misalignment error is weighted by the gradient mean $\sum_{d=1}^D{m_d^2(t)}$.
In the special case when the gradient has zero mean, the MSE in (\ref{equ:MSE}) reduces to that in \cite{cao2019optimal} where the composite misalignment error is absent.
Our objective is to minimize MSE in (\ref{equ:MSE}), by jointly optimizing the transmit power $p_k(t)$ at all devices and the denoising factor $\eta(t)$ at the edge server, subject to the individual power budget in (\ref{equ:power_constrain}) at each iteration $t$.

\subsection{Gradient Statistics}
In general, the individual misalignment error and the composite misalignment error in MSE of the gradient aggregation (\ref{equ:MSE}) cannot be forced to zero simultaneously due to the peak power budget on each device.
It is difficult to capture the tradeoff between the two errors by directly using their respective weights, namely, the gradient variance and gradient mean.
To tackle this issue, we introduce two alternative parameters of gradient statistics in this subsection.

Let $\alpha(t)$ denote the mean squared norm (MSN) of $\bm{g}(t)$, i.e., $\mathbb{E}[\|\bm{g}(t)\|^2]$, which is given by
\begin{equation}
\alpha(t)=\sum_{d=1}^D{\bigg(\sigma_d^2(t)+m_d^2(t)\bigg)},
\label{equ:alpha}
\end{equation}
and let $\beta(t)$ denote the squared multivariate coefficient of variation (SMCV) of $\bm{g}(t)$, which is given by
\begin{equation}
\beta(t)=\frac{\sum_{d=1}^D{\sigma_d^2(t)}}{\sum_{d=1}^D{m_d^2(t)}}.
\label{equ:beta}
\end{equation}
While $\alpha(t)$ measures the average absolute value of the gradient values at iteration $t$, $\beta(t)$ is a measure of relative dispersion of the gradient at iteration $t$.
In Section III-B, we shall show that the optimal transmit power is dominated by $\beta(t)$ and independent of $\alpha(t)$.
Fig.~\ref{fig:alpha_beta} illustrates the experimental results of the alternative gradient statistics $\alpha(t)$ and $\beta(t)$ of three datasets, MNIST, CIFAR-10, and SVHN, where the gradients are updated ideally without any transmission error.
Both IID and non-IID partitions are considered for the training dataset and the generation of IID and non-IID data distributions is specified in Section V.
Each value of $\alpha(t)$ and $\beta(t)$ is obtained by averaging over 300 model trainings.
It is observed that as the training time increases, the gradient MSN $\alpha(t)$ gradually decreases and approaches to a constant while the gradient SMCV $\beta(t)$ keeps increasing from a small starting point for all the three datasets.
Intuitively, in SGD-based learning, the gradient mean $\{m_d(t)\}$ is large at the begining, then gradually approaches to zero when the model converges; the gradient variance $\{\sigma_d^2(t)\}$, on the other hand, remains approximately unchanged due to the randomness of local datasets throughout the training process.
As a result, according to the definitions in (\ref{equ:alpha}) and (\ref{equ:beta}), the absolute squared norm of the gradient, $\alpha(t)$, decreases during the training and converges to a constant , but the relative dispersion of the gradient, $\beta(t)$, is very small when the training just begins and then gradually increases when the training continues.
It is also observed that $\alpha(t)$ and $\beta(t)$ in non-IID partition are much larger than those in IID partition for all the three datasets.
This indicates that the gradient distribution with non-IID dataset partition is more dispersive than that with IID dataset partition as expected.
Note that the gradient MSN $\alpha(t)$ of CIFAR-10 and SVHN datasets are larger than MNIST dataset.
This is because the dimensions of model parameters for training CIFAR-10 and SVHN datasets are larger than MNIST dataset.

\begin{figure}[htb]
\centering

\subfigure[MNIST dataset.]
{\begin{minipage}[t]{0.32\linewidth}
\centering
\includegraphics[width=2.2in]{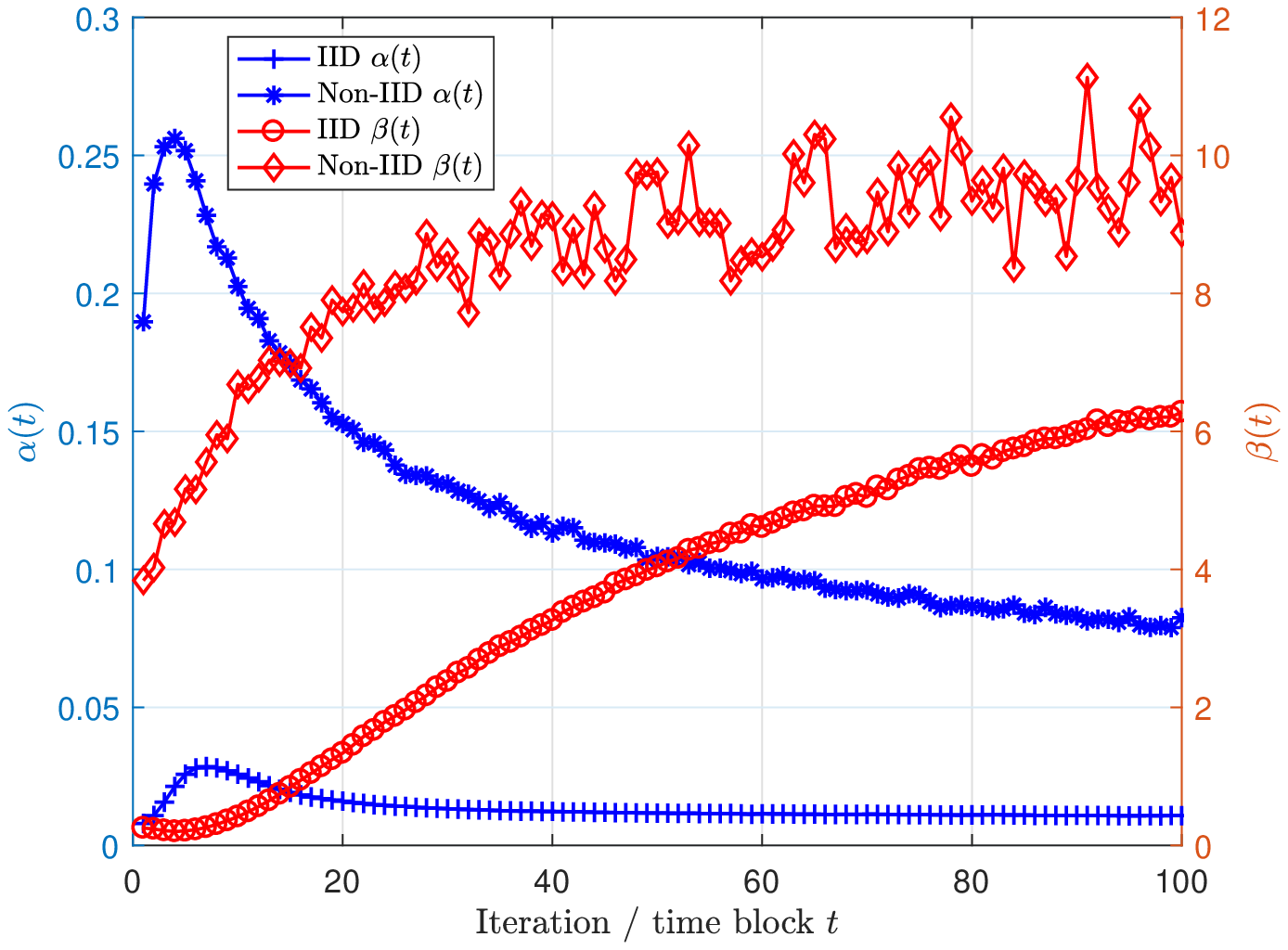}
\label{fig:alpha_beta_mnist}
\end{minipage}}
\subfigure[CIFAR-10 dataset.]
{\begin{minipage}[t]{0.32\linewidth}
\centering
\includegraphics[width=2.2in]{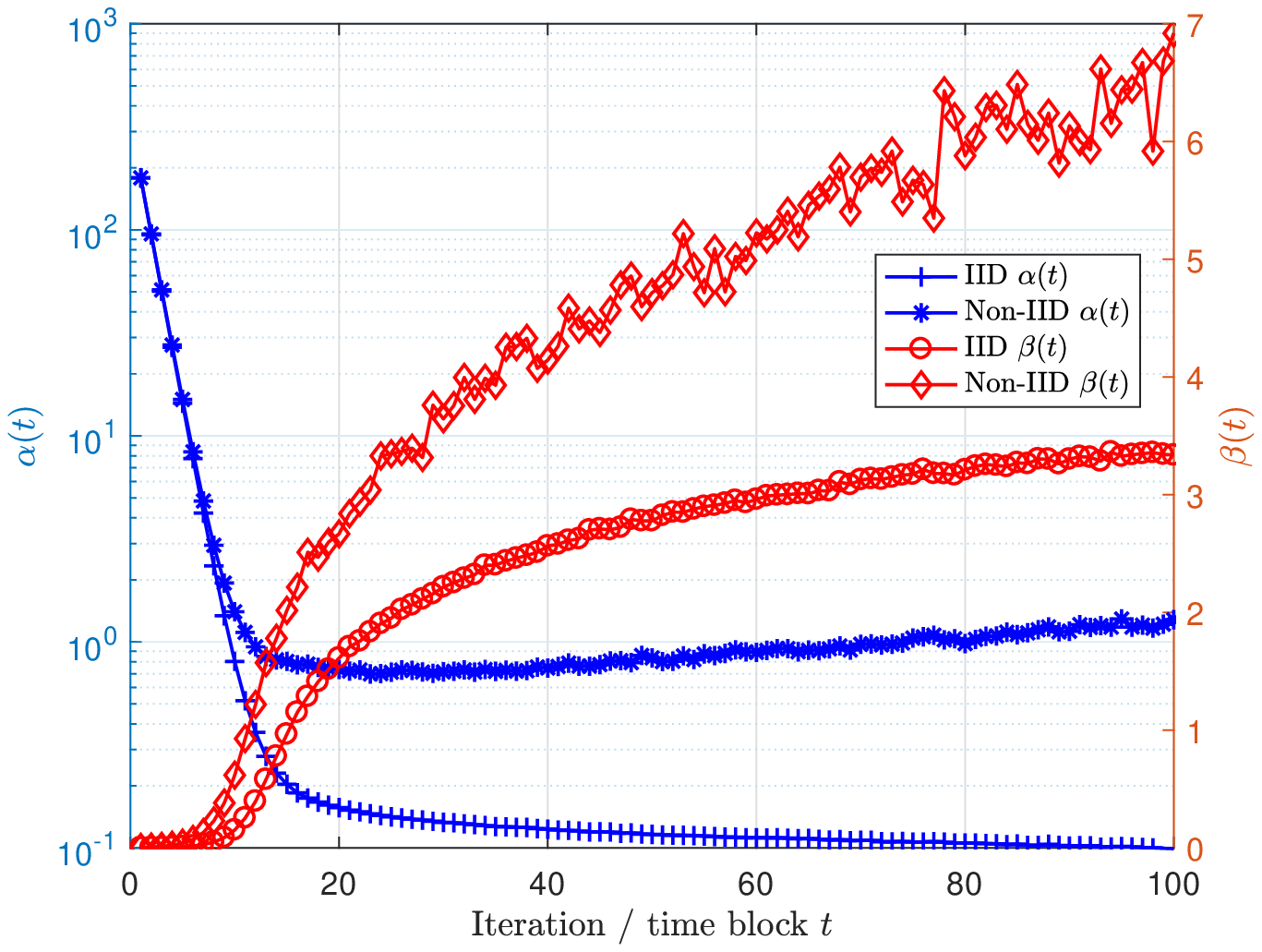}
\label{fig:alpha_beta_cifar}
\end{minipage}}
\subfigure[SVHN dataset.]
{\begin{minipage}[t]{0.32\linewidth}
\centering
\includegraphics[width=2.2in]{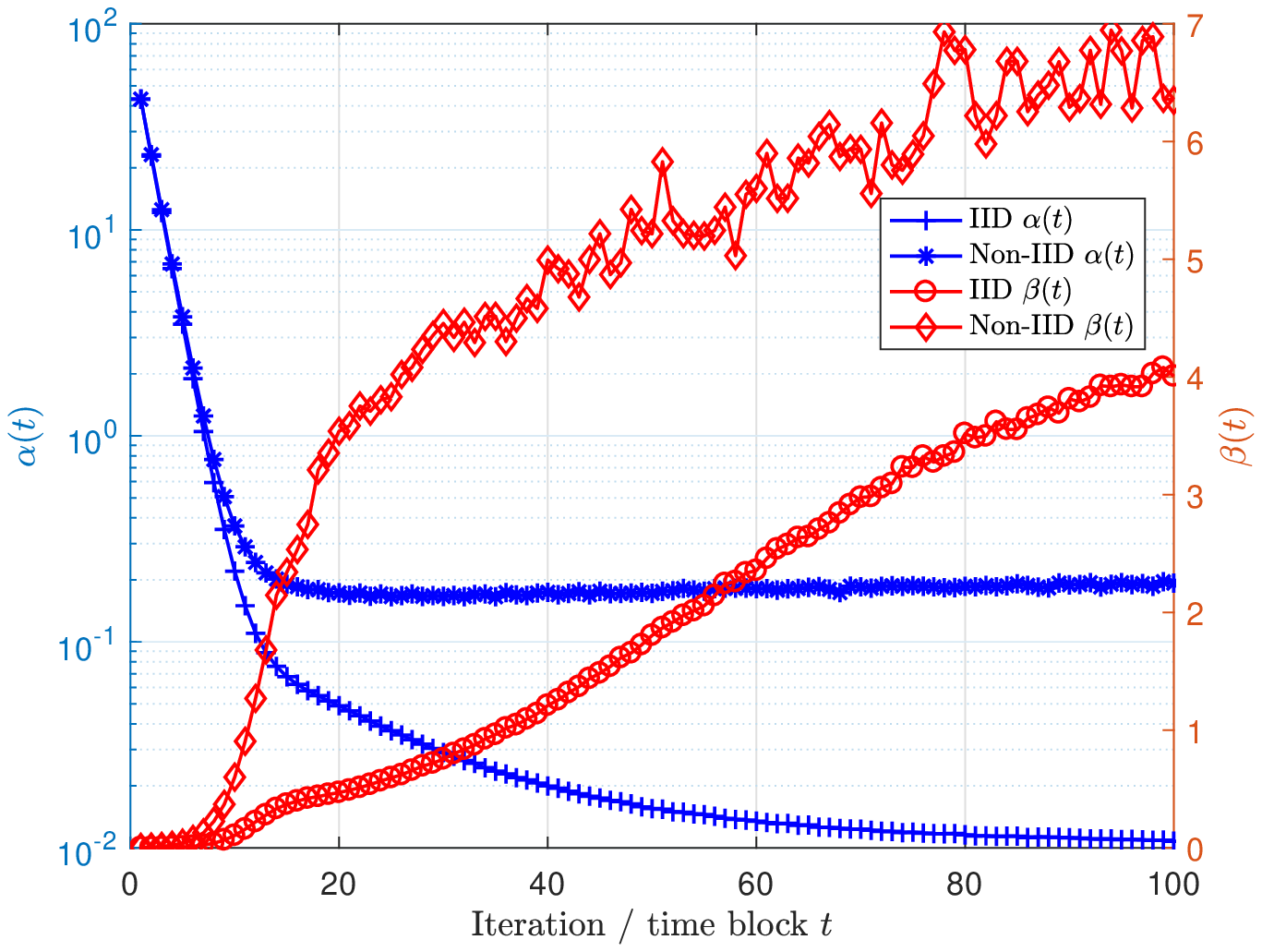}
\label{fig:alpha_beta_SVHN}
\end{minipage}}

\centering
\caption{Experimental results of MSN $\alpha(t)$ (left y-axis in linear scale) and SMCV $\beta(t)$ (right y-axis) over iterations for three datasets where the number of edge devices is 10 and the local mini-batch size is 20.}
\label{fig:alpha_beta} 
\end{figure}

By (\ref{equ:alpha}) and (\ref{equ:beta}), the MSE in (\ref{equ:MSE}) can be rewritten as (omitting the constant coefficient $1/K^2$ for convenience)
\begin{align}
\text{MSE}(t)=&\frac{\beta(t)\alpha(t)}{\beta(t)+1}\sum_{k\in\mathcal{K}}{\left(\frac{\sqrt{p_k(t)}|h_k(t)|}{\sqrt{\eta(t)\alpha(t)}}-1\right)^2}\nonumber\\
&+\frac{\alpha(t)}{\beta(t)+1}{\left(\frac{\sum_{k\in\mathcal{K}}\sqrt{p_k(t)}|h_k(t)|}{\sqrt{\eta(t)\alpha(t)}}-K\right)^2}+\frac{D\sigma_n^2}{\eta(t)}.
\label{equ:MSE_reformulation}
\end{align}
It is seen from (\ref{equ:MSE_reformulation}) that while the gradient MSN $\alpha(t)$ appears linearly in the weights of both individual and composite misalignment errors, the gradient SMCV $\beta(t)$ plays a more distinguishing role in the MSE expression.
In particular, when $\beta(t)\rightarrow 0$, which could be the case when the model training just begins (as could be verified by Fig.~\ref{fig:alpha_beta}, the individual signal misalignment error can be neglected.
When $\beta(t)\rightarrow\infty$, which could be the case when the model training converges or during the middle of the training when the dataset is highly non-IID (as could be verified by Fig.~\ref{fig:alpha_beta}, the composite signal misalignment error disappears.

\section{Optimal Power Control with Known Gradient Statistics}
In this section, we formulate and solve the optimal power control problem for minimizing MSE when the gradient statistics $\alpha(t)$ and $\beta(t)$ are known.
For convenience, we omit iteration index $t$ in this section.
For each device $k\in\mathcal{K}$, we define its \emph{aggregation level} with power $p$ and denoising factor $\eta$ as
\begin{equation}
G_k(p,\eta)=\sqrt{\frac{p}{\eta\alpha}}|h_k|,
\end{equation}
which represents the weight of the gradient from device $k$ in the global gradient aggregation (\ref{recover_gradient})\footnote{The weight should be 1 for all devices in the ideal case.}.
Furthermore, we define \emph{aggregation capability} of device $k$ as its aggregation level with peak power $P_k$ and unit denoising factor $\eta=1$, i.e., $C_k=G_k(P_k,1)=\sqrt{\frac{P_k}{\alpha}}|h_k|$.
Without loss of generality, we assume that
\begin{equation}
C_1\leq...\leq C_k\leq...\leq C_K.
\label{equ:sort}
\end{equation}

\subsection{Power Control Problem for General Case}
In this subsection, we consider the optimal power control problem for MSE minimization in the general case.
The problem is formulated as
\begin{subequations}
\begin{align}
\mathcal{P}_1:\quad\min\quad&\frac{\beta\alpha}{\beta+1}\sum_{k\in\mathcal{K}}{\left(G_k(p_k,\eta)-1\right)^2}+\frac{\alpha}{\beta+1}{\left(\sum_{k\in\mathcal{K}}G_k(p_k,\eta)-K\right)^2}+\frac{D\sigma_n^2}{\eta} \label{equ:problem_formulation}\\
\mathnormal{s.t.}\quad &~~~~~~~~~~~~~~~~~~~~0\leq p_k\leq P_k,~~\forall k\in\mathcal{K}\\
&~~~~~~~~~~~~~~~~~~~~\eta\geq 0.
\label{equ:limit_BS}
\end{align}
\end{subequations}
Different from the power control problem in \cite{cao2019optimal}, the objective function in (\ref{equ:problem_formulation}) contains not only the individual misalignment error ($\frac{\beta\alpha}{\beta+1}\sum_{k\in\mathcal{K}}{(G_k(p_k,\eta)-1)^2}$), but also the composite misalignment error  ($\frac{\alpha}{\beta+1}{(\sum_{k\in\mathcal{K}}G_k(p_k,\eta)-K)^2}$).
Problem $\mathcal{P}_1$ is non-convex in general.
Even if the denoising factor $\eta$ is given, problem $\mathcal{P}_1$ is still hard to solve due to the coupling of each transmit power $p_k$.
In the following, we present some properties of the optimal solution.

\begin{lemma}
\label{lemma:eta_lower_bound}
The optimal denoising factor $\eta^*$ for problem $\mathcal{P}_1$ satisfies $\eta^*\geq C_1^2$.
\end{lemma}
\begin{proof}
Please refer to Appendix \ref{proof:lemma:eta_lower_bound}.
\end{proof}

Lemma \ref{lemma:eta_lower_bound} reduces the range of $\eta^*$ and shows that the optimal transmit power of the $1$-st device is $p^*_1=P_1$.
\begin{lemma}
\label{lemma:optimal_power_constrain}
The optimal power control policy satisfies $p^*_k=P_k,\forall{k}\in\{1,...,l\}$ and $p^*_k<P_k,\forall{k}\in\{l+1,...,K\}$ for some $l\in\mathcal{K}$.
\end{lemma}
\begin{proof}
Please refer to Appendix \ref{proof:lemma:optimal_power_constrain}.
\end{proof}

Based on Lemma \ref{lemma:optimal_power_constrain}, solving problem $\mathcal{P}_1$ can be equivalent to minimizing the objective function in the following $K$ exclusive subregions of the global power region, denoted as $\{\mathcal{M}_l\}_{l\in\mathcal{K}}$ and comparing their corresponding optimal solutions to obtain the global optimal solution:
\begin{align}
\mathcal{M}_l=\bigg\{[p_1,\cdots,p_K]\in\mathbb{R}^K|&p_k=P_k,\forall{k}\in\{1,...,l\},0\leq p_k<P_k,\forall{k}\in\{l+1,...,K\}\bigg\}.
\label{equ:domain}
\end{align}
To facilitate the derivation, we denote $\tilde{\mathcal{M}_l}$ as a relaxed region of $\mathcal{M}_l$ by removing the condition $p_k<P_k$, for $k\in\{l+1, ...,K\}$, i.e.,
\begin{align}
\tilde{\mathcal{M}_l}=\bigg\{[p_1,\cdots,p_K]\in\mathbb{R}^K|&p_k=P_k,\forall{k}\in\{1,...,l\},p_k\geq 0,\forall{k}\in\{l+1,...,K\}\bigg\}.
\label{equ:domain}
\end{align}
For the sub-problem defined in each relaxed subregion $\tilde{\mathcal{M}_l}$, for $l\in\mathcal{K}$, taking the derivative of the objective function (\ref{equ:problem_formulation}) w.r.t. $p_k$ and equating it to zero for all $k\in\{l+1,...,K\}$, we obtain the optimal transmit power $\tilde{p}^*_{l,k}$ at any given $\eta$ as
\begin{equation}
\tilde{p}^*_{l,k}=\left[\frac{\beta+K-\sum_{i=1}^{l}{G_i(P_{i},\eta)}}{\beta+K-l}\right]^2\cdot\frac{\alpha\eta}{|h_k|^2}, k\in\{l+1,...,K\}.
\label{equ:power_control_subregion}
\end{equation}
Note that by such power control in (\ref{equ:power_control_subregion}), the aggregation level $G_k(\tilde{p}^*_{l,k},\eta)$ of each device $k\in\{l+1,...,K\}$ is the same, given by
\begin{align}
G_0(l)=\frac{\beta+K-\frac{1}{\sqrt{\eta}}\sum_{i=1}^l{C_i}}{\beta+K-l}.
\label{equ:equal_aggregation_level}
\end{align}
Substituting (\ref{equ:power_control_subregion}) back to (\ref{equ:problem_formulation}) and letting its derivative w.r.t. $\eta$ be zero, we derive the optimal denoising factor $\eta$ in closed-form for problem $\mathcal{P}_1$ defined in the $l$-th relaxed subregion, i.e.,
\begin{align}
\sqrt{\tilde{\eta}^*_l}=\frac{\frac{\beta\alpha}{\beta+1}\sum\limits_{i=1}^{l}C_i^2+\frac{\beta\alpha}{(\beta+K-l)(\beta+1)}\bigg(\sum\limits_{i=1}^{l}C_i\bigg)^2+D\sigma_n^2}{\frac{\beta(\beta+K)\alpha}{(\beta+K-l)(\beta+1)}\sum\limits_{i=1}^{l}C_i}.
\label{equ:eta_function_subregion}
\end{align}
Note that $\tilde{p}^*_{l,k}$ may be not less than its power constraint $P_k$ for some $k\in\{l+1,...,K\}$ and thus the corresponding $\bm{\tilde{p}}^*_l$ may not lie in the subregion $\mathcal{M}_l$.
If this happens, the optimal transmit power of the sub-problem defined in subregion $\mathcal{M}_l$ is irrelevant and does not need to be considered.
This is given in the following lemma.
\begin{lemma}
\label{lemma:illegal_solution}
For the power control problem $\mathcal{P}_1$ defined in the $l$-th relaxed subregion $\tilde{\mathcal{M}_l}$, if $\exists{k}\in\{l+1,...,K\}$ such that $\tilde{p}^*_{l,k}\geq P_k$, the global optimal power $\bm{p}^*\triangleq [p^*_1,...p^*_K]$ of problem $\mathcal{P}_1$ must not be in $\mathcal{M}_l$.
\end{lemma}
\begin{proof}
Please refer to Appendix \ref{proof:lemma:illegal_solution}.
\end{proof}

Lemma \ref{lemma:illegal_solution} shows that only the transmit power vectors $\tilde{\bm{p}}^*_l$'s with elements satisfying $\tilde{p}^*_{l,k}<P_k,\forall{k}\in\{l+1,...,K\}$ are legal candidates of problem $\mathcal{P}_1$.
Let $\mathcal{L}$ denote the index set of the corresponding relaxed subregions.
Note that $\mathcal{L}$ is non-empty because $\tilde{\bm{p}}^*_K$ is always a legal transmit power candidate.
Then, we only need to compare the legal candidate values to obtain the minimum MSE
\begin{equation}
l^*=\arg\min_{l\in\mathcal{L}}V_{l},
\label{equ:optimal_tau}
\end{equation}
where $V_l$ is the optimal value of (\ref{equ:problem_formulation}) in relaxed subregion $\tilde{\mathcal{M}_l}$ of $\mathcal{P}_1$ and can be easily obtained by substituting (\ref{equ:power_control_subregion}) and (\ref{equ:eta_function_subregion}) back to (\ref{equ:problem_formulation}).
The optimal solution to problem $\mathcal{P}_1$ is given as follows.
\begin{theorem}
\label{theorem:optimality}
The optimal transmit power at each device that solves problem $\mathcal{P}_1$ is given by
\begin{equation}
p_k^*=
\left\{
             \begin{array}{ll}
             P_k, &\forall{k}\in\{1,...,l^*\}\\
             \left[\frac{\beta+K-\sum_{i=1}^{l^*}{G_i(P_{i},\eta^*)}}{\beta+K-l^*}\right]^2\cdot\frac{\alpha\eta^*}{|h_k|^2}, &\forall{k}\in\{l^*+1,...,K\},
             \end{array}
\right.
\label{equ:power_control}
\end{equation}
and the optimal denoising factor at the edge server is given by
\begin{align}
\sqrt{\eta^*}=\frac{\frac{\beta\alpha}{\beta+1}\sum\limits_{i=1}^{l^*}C_i^2+\frac{\beta\alpha}{(\beta+K-l^*)(\beta+1)}\bigg(\sum\limits_{i=1}^{l^*}C_i\bigg)^2+D\sigma_n^2}{\frac{\beta(\beta+K)\alpha}{(\beta+K-l^*)(\beta+1)}\sum\limits_{i=1}^{l^*}C_i},
\label{equ:eta_function}
\end{align}
where $l^*$ is given in (\ref{equ:optimal_tau}).
\end{theorem}
\begin{proof}
Please refer to Appendix \ref{proof:theorem:optimality}.
\end{proof}
\begin{remark}
Theorem \ref{theorem:optimality} shows that these devices $k\in\{1,...,l^*\}$ with aggregation capability not higher than that of device $l^*$ should transmit their gradients with full power, i.e., $p_k=P_k$, while devices $k\in\{l^*+1,...,K\}$ with aggregation capability higher than that of device $l^*$ should transmit with the power so that they have the same aggregation level, given by
\begin{align}
G^*_0=\frac{\beta+K-\sum_{i=1}^{l^*}{G_i(P_{i},\eta^*)}}{\beta+K-l^*},
\label{equ:equal_aggregation}
\end{align}
somewhat analogous to channel inversion.
\end{remark}

\subsection{On The Optimal Transmit Power Function}
In this subsection, we analyse the optimal transmit power $\bm{p}^*$ as a vector function of the gradient MSN $\alpha$, the gradient SMCV $\beta$ and the noise variance $\sigma^2_n$, i.e., $\bm{p}^*(\alpha,\beta,\sigma^2_n)=[p^*_1(\alpha,\beta,\sigma^2_n),...,p^*_K(\alpha,\beta,\sigma^2_n)]$.
Note that the vector function $\bm{p}^*(\alpha,\beta,\sigma^2_n)$ might have abrupt changes at some $(\alpha,\beta,\sigma^2_n)$ due to the discrete nature of the optimal device threshold $l^*$ in Theorem \ref{theorem:optimality}.
In this work, however, we can show that $\bm{p}^*(\alpha,\beta,\sigma^2_n)$ is continuous\footnote{$\bm{p}^*(\alpha,\beta,\sigma^2_n)$ is a continuous vector-valued function if and only if each element $p^*_k(\alpha,\beta,\sigma^2_n)$, for $k\in\{1,...,K\}$ is a continuous function.} everywhere for all $(\alpha\geq 0,\beta\geq 0,\sigma^2_n\geq 0)$.
Denote $\mathcal{X}_l\subseteq\mathbb{R}^3$ as the domain of the vector function $\bm{p}^*(\alpha,\beta,\sigma^2_n)$ when $\bm{p}^*(\alpha,\beta,\sigma^2_n)$ lies in the subregion $\mathcal{M}_l$, for $l\in\mathcal{K}$.
That is, $\forall (\alpha,\beta,\sigma^2_n)\in\mathcal{X}_l$, one can have $l^*(\alpha,\beta,\sigma^2_n)=l$.

We first show that $\bm{p}^*(\alpha,\beta,\sigma^2_n)$ is continuous and monotonic within each domain $\mathcal{X}_l$, $\forall l\in\mathcal{K}$.
Let $\bm{\tilde{p}}^*_l(\alpha,\beta,\sigma^2_n)$ denote the optimal power of the sub-problem defined in the $l$-th relaxed subregion $\tilde{\mathcal{M}_l}$, $\forall l\in \mathcal{K}$.
Based on (\ref{equ:power_control_subregion}) and (\ref{equ:eta_function_subregion}), it is obvious that $\bm{\tilde{p}}^*_l(\alpha,\beta,\sigma^2_n)$ is continuous within each domain $\mathcal{X}_l$.
Moreover, $\bm{\tilde{p}}^*_l(\alpha,\beta,\sigma^2_n)$ increases monotonically with the noise variance $\sigma_n^2$ since enlarging $\sigma_n^2$ increases both $\tilde{\eta}^*_l$ and $\tilde{p}_{l,k}^*$, for $k\in\{l+1,...,K\}$.
Furthermore, $\bm{\tilde{p}}^*_l(\alpha,\beta,\sigma^2_n)$ is constant w.r.t. the gradient MSN $\alpha$ since it can be easily shown that at the optimal solution $\alpha\tilde{\eta}^*_l(\alpha,\beta,\sigma^2_n)$ is constant w.r.t. $\alpha$ by substituting $C_k=\sqrt{\frac{P_k}{\alpha}}|h_k|$ to (\ref{equ:eta_function_subregion}).
In addition, $\bm{\tilde{p}}^*_l(\alpha,\beta,\sigma^2_n)$ decreases monotonically with the gradient SMCV $\beta$ since it can be easily shown that $\frac{\partial\bm{\tilde{p}}^*_l(\alpha,\beta,\sigma^2_n)}{\partial\beta}$ is always negative.
Thus $\bm{\tilde{p}}^*_l(\alpha,\beta,\sigma^2_n)$ is monotonic within each domain $\mathcal{X}_l$, $l\in\mathcal{K}$.
By definition, within each domain $\mathcal{X}_l$, $\bm{p}^*(\alpha,\beta,\sigma^2_n)$ is equal to $\bm{\tilde{p}}^*_l(\alpha,\beta,\sigma^2_n)$, thus, the vector function $\bm{p}^*(\alpha,\beta,\sigma^2_n)$ is also continuous and monotonic within each domain $\mathcal{X}_l$, for $l\in \mathcal{K}$.

Then we find the boundary of each domain $\mathcal{X}_l$, for $l\in\mathcal{K}$ and the corresponding optimal transmit power $\bm{p}^*$.
To this end, we need the following lemma on the lower and upper bounds of the optimal transmit power values.
\begin{lemma}
\label{lemma:condition_2_optimal}
The optimal transmit power $\bm{p}^*$ lies in subregion $\mathcal{M}_l$ if and only if the optimal transmit power $\bm{\tilde{p}}^*_l$ in the $l$-th relaxed subregion $\mathcal{\tilde M}_l$ satisfies: $(\frac{C_{l}\sqrt{\alpha}}{|h_k|})^2\leq\tilde{p}^*_{l,k}<(\frac{C_{l+1}\sqrt{\alpha}}{|h_k|})^2,\forall k\in\{l+1,...,K\}$.
\end{lemma}
\begin{proof}
Please refer to Appendix \ref{proof:lemma:condition_2_optimal}.
\end{proof}

Lemma \ref{lemma:condition_2_optimal} shows that in each domain $\mathcal{X}_l,l\in\mathcal{K}$, the range of $\bm{p}^*(\alpha,\beta,\sigma^2_n)$ is left-closed and right-open intervals in $\left[(\frac{C_{l}\sqrt{\alpha}}{|h_k|})^2,(\frac{C_{l+1}\sqrt{\alpha}}{|h_k|})^2\right),\forall k\in\{l+1,...,K\}$.
Recall that $\bm{\tilde{p}}^*_l(\alpha,\beta,\sigma^2_n)$ is continuous and monotonic.
The optimal transmit power $\bm{p}^*(\alpha,\beta,\sigma^2_n)=\bm{p}^*_l(\alpha,\beta,\sigma^2_n)$ sits on the lower bound of the range when $(\alpha,\beta,\sigma^2_n)$ is at the lower boundary of domain $\mathcal{X}_l$, denoted as $\mathcal{L}_l\triangleq \left\{(\alpha,\beta,\sigma^2_n)|\tilde{p}^*_{l,k}(\alpha,\beta,\sigma^2_n)=(\frac{C_{l}\sqrt{\alpha}}{|h_k|})^2,\forall k\in\{l+1,...,K\}\right\}$, and $\bm{p}^*(\alpha,\beta,\sigma^2_n)=\bm{p}^*_l(\alpha,\beta,\sigma^2_n)$ approaches the upper bound of the range when $(\alpha,\beta,\sigma^2_n)$ approaches the upper boundary of domain $\mathcal{X}_l$, denoted as $\mathcal{U}_l\triangleq \left\{(\alpha,\beta,\sigma^2_n)|\tilde{p}^*_{l,k}(\alpha,\beta,\sigma^2_n)=(\frac{C_{l+1}\sqrt{\alpha}}{|h_k|})^2,\forall k\in\{l+1,...,K\}\right\}$.

Next, we consider the continuity of $\bm{p}^*(\alpha,\beta,\sigma^2_n)$ at boundaries $\mathcal{L}_l$ and $\mathcal{U}_l$ for each $l\in\mathcal{K}$ in the following lemma.
\begin{lemma}
\label{lemma:continuous_globally}
The optimal transmit power function $\bm{p}^*(\alpha,\beta,\sigma^2_n)$ is continuous at $\mathcal{U}_l=\mathcal{L}_{l+1}$, for $l\in\{1,...,K-1\}$.
\end{lemma}
\begin{proof}
Please refer to Appendix \ref{proof:lemma:continuous_globally}.
\end{proof}

Finally, we can formally conclude the following property of the optimal transmit power function with respect to the gradient statistics and noise variance.
\begin{property}
The optimal transmit power $\bm{p}^*(\alpha,\beta,\sigma^2_n)$ of problem $\mathcal{P}_1$ is a continuous and monotonic vector function in the entire domain of $(\alpha\geq 0,\beta\geq 0,\sigma^2_n\geq 0)$.
Moreover, it decreases with $\beta$, increases with $\sigma_n^2$, and remain unchanged w.r.t. $\alpha$.
\end{property}

\begin{figure}[htb]
\centering

\subfigure[Average received SNR=0 dB.]
{\begin{minipage}[t]{0.48\linewidth}
\centering
\includegraphics[width=3.2in]{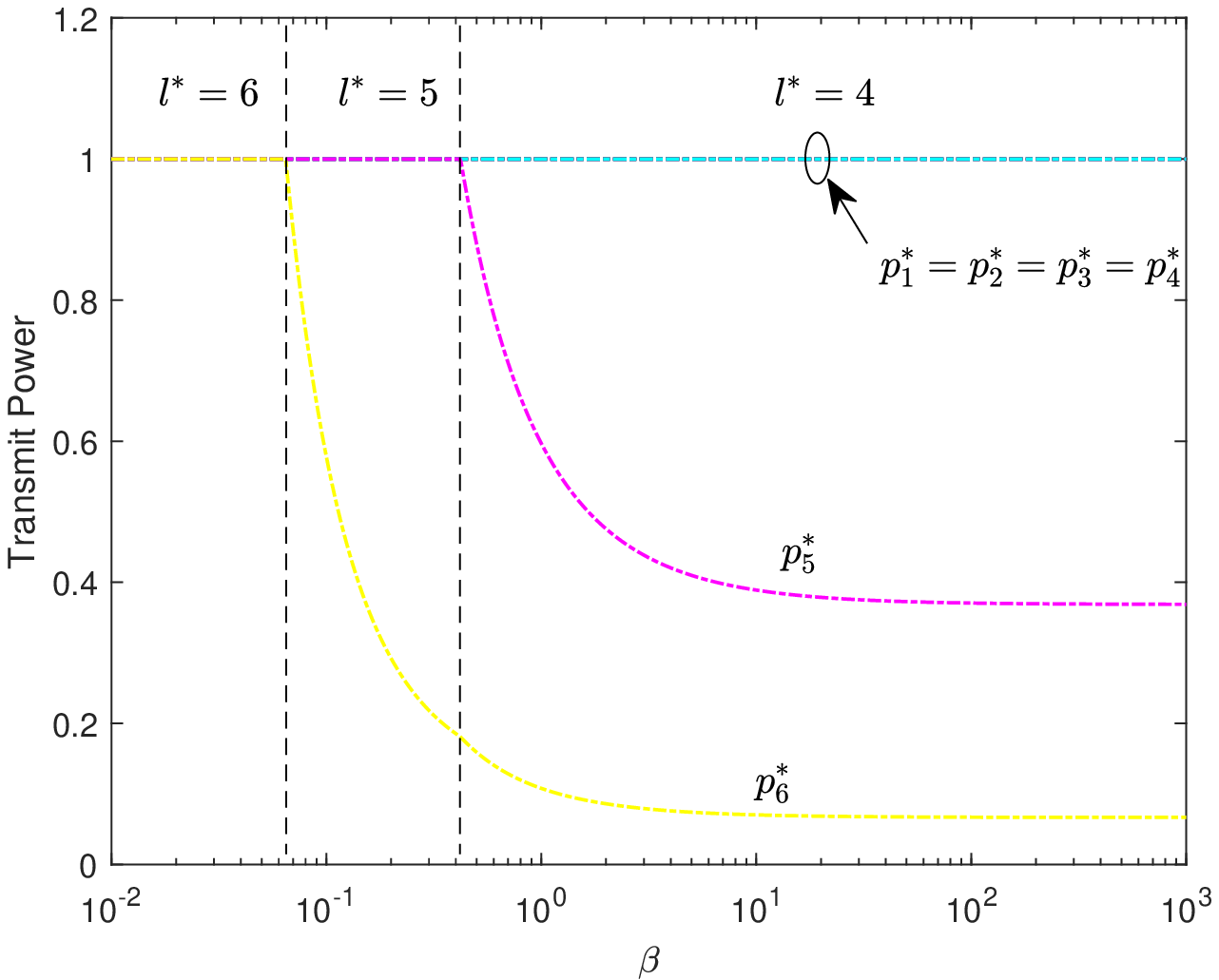}
\label{fig:p_beta_0db}
\end{minipage}}
\subfigure[Average received SNR=10 dB.]
{\begin{minipage}[t]{0.48\linewidth}
\centering
\includegraphics[width=3.2in]{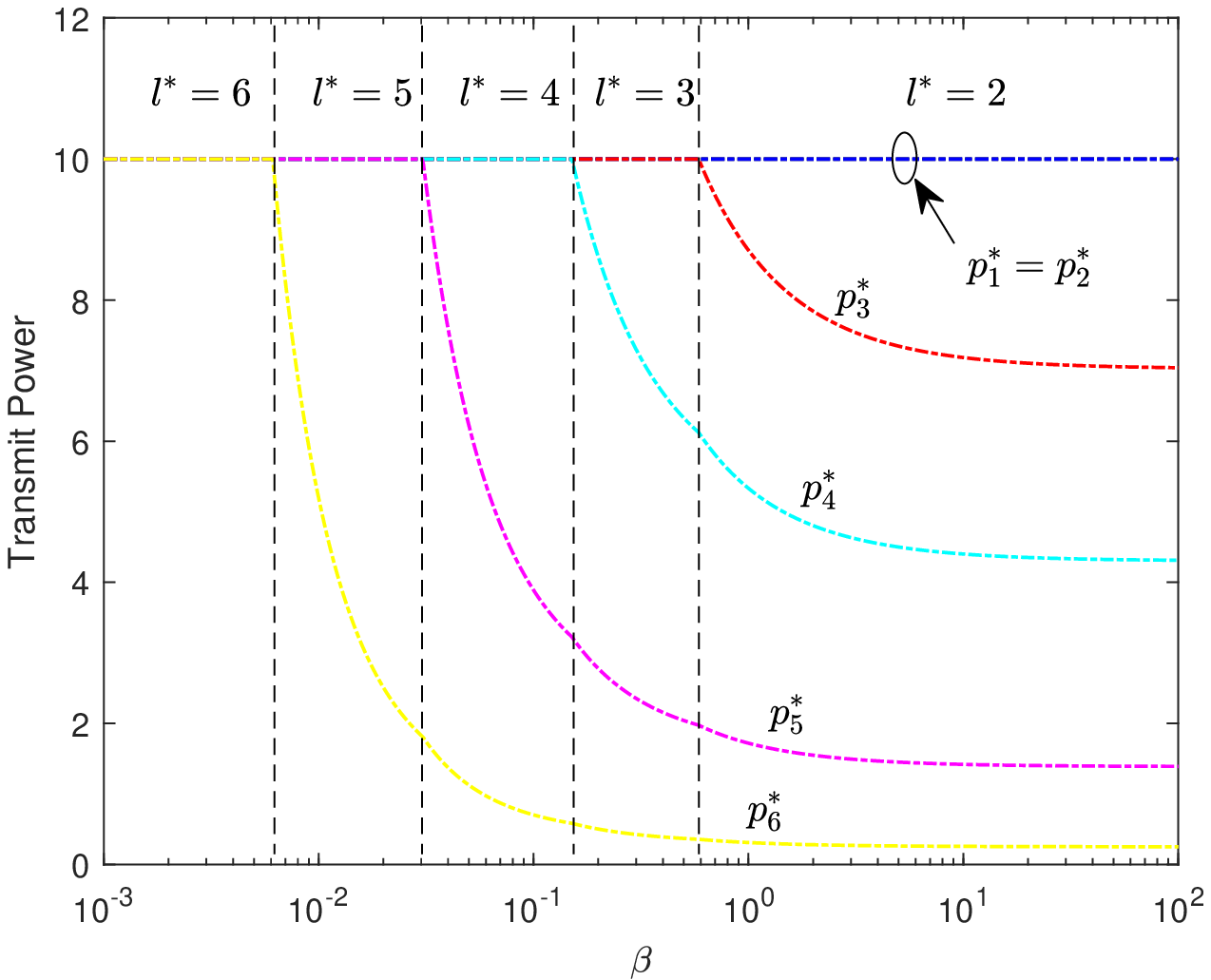}
\label{fig:p_beta_10db}
\end{minipage}}

\centering
\caption{Illustration of the optimal transmit power $\bm{p}^*$ as a function of gradient SMCV $\beta$.}
\label{fig:optimal power control} 
\end{figure}
Take a system with $K = 6$ devices for illustration.
Fig.~\ref{fig:optimal power control} shows the optimal transmit power of each device $p^*_k$ as a function of gradient SMCV $\beta$, together with the corresponding optimal index $l^*$.
The results are based on one channel realization of each device, $|h_k|$, taken independently from normalized Rayleigh distribution, given by $[0.50,0.82,0.85,1.16,2.09,2.83]$.
The peak power constraint of each device is set to be same, resulting in the same average received SNR, given by $\frac{P_k}{D\sigma_n^2}=\mbox{5 dB and 10 dB}$, respectively.
The gradient MSN $\alpha=0.25$, noise variance $\sigma^2=1$ and $D=1$.
Fig.~\ref{fig:optimal power control} clearly verifies that the optimal transmit power $\bm{p}^*(\alpha,\beta,\sigma^2_n)$ is a continuous and monotonically decreasing vector function of the gradient SMCV $\beta$.
In particular, it is seen that when $\beta\rightarrow 0$, all the devices transmit with full power; when beta increases, the power of the devices with large aggregation capability decreases, then gradually approaches a constant, and the larger the aggregation capability is, the smaller the constant transmit power is.
In addition, it is observed from Fig.~\ref{fig:optimal power control} that when the peak power budget increases, more devices can transmit with less power than its peak value.
Equivalently, the optimal transmit power decreases when the noise variance decreases.

\subsection{Power Control Problem for Special Cases}
In this subsection, we show that the threshold-based power control in \cite{cao2019optimal} and the full power transmission are two special cases of the optimal power control policy where the gradient SMCV $\beta\rightarrow\infty$ and $\beta\rightarrow 0$, respectively.
\subsubsection{$\beta\rightarrow\infty$}
As discussed in Section II-D, this case may happen when the model training converges and/or the dataset is highly non-IID.
In this case, the problem $\mathcal{P}_1$ reduces to the power control problem in \cite{cao2019optimal}.
The optimal solution when $\beta\rightarrow\infty$ can be reduced directly from the solution of problem $\mathcal{P}_1$.
Specifically, by letting $\beta\rightarrow\infty$ in Theorem \ref{theorem:optimality}, we have the following Corollary,
\begin{corollary}
\label{corollary:optimal_infinity}
The optimal transmit power when $\beta\rightarrow\infty$ has a threshold-based structure, i.e.,
\begin{equation}
p_k^*=
\left\{
             \begin{array}{ll}
             P_k, &\forall{k}\in\{1,...,l^*\}\\
             \frac{\alpha\eta^*}{|h_k|^2}, &\forall{k}\in\{l^*+1,...,K\},
             \end{array}
\right.
\label{equ:infinity_power_control}
\end{equation}
where the optimal denoising factor is given as
\begin{align}
\eta^*=\bigg(\frac{\alpha\sum_{i=1}^{l^*}{C_i^2}+D\sigma_n^2}{\alpha\sum_{i=1}^{l^*}{C_i}}\bigg)^2,
\label{equ:infinity_denoising_factor}
\end{align}
and $l^*$ is given in (\ref{equ:optimal_tau}).
\end{corollary}

Note that the optimal denoising factor $\eta^*$ when $\beta\rightarrow\infty$ can be interpreted as the threshold, since whether a device transmits with full power or not depends entirely on the comparison between its aggregation capability $C_k$ and $\sqrt{\eta^*}$.
Nevertheless, such threshold interpretation of $\eta^*$ does not hold for problem $\mathcal{P}_1$ with general $\beta$.

\subsubsection{$\beta\rightarrow 0$}
As discussed in Section II-D, $\beta\rightarrow 0$ could happen when the model training just begins and the dataset is IID.
In this case, the individual signal misalignment error disappears in the MSE expressions.
The original MSE expression reduces to $\alpha{\left(\sum_{k\in\mathcal{K}}G_k(p_k,\eta)-K\right)^2}+\frac{D\sigma_n^2}{\eta}$.
Note that the optimal solution when $\beta\rightarrow 0$ cannot be reduced directly from Theorem \ref{theorem:optimality} by simply letting $\beta\rightarrow 0$ as $l^*$ is unknown.

\begin{corollary}
\label{corollary:beta_zero}
The optimal transmit power when $\beta\rightarrow 0$ is full power transmission, i.e.,
\begin{align}
p_k^*=P_k,~~~\forall{k}\in\mathcal{K},
\label{equ:zero_power_control}
\end{align}
and the optimal denoising factor is given by
\begin{align}
\eta^*=\bigg(\frac{\alpha\big(\sum_{i\in\mathcal{K}}{C_i}\big)^2+D\sigma_n^2}{\alpha K\sum_{i\in\mathcal{K}}{C_i}}\bigg)^2.
\label{equ:zero_denoising_factor}
\end{align}
\end{corollary}
\begin{proof}
It is obvious that the optimal denoising factor satisfies $\eta^*\geq\sum_{k\in\mathcal{K}}{C_k}$.
For any denoising factor $\eta\geq\frac{1}{K^2}\left(\sum_{k\in\mathcal{K}}{C_k}\right)^2$, it must hold that $\sum_{k\in\mathcal{K}}G_k(p_k,\eta)\leq K$ when $\beta\rightarrow 0$.
Therefore, for minimizing the composite signal misalignment error $\left(\sum_{k\in\mathcal{K}}G_k(p_k,\eta)-K\right)^2$, all the devices should transmit with full power, i.e., $p_k^*=P_k,\forall k\in\mathcal{K}$.
The MSE expression $\alpha{\left(\sum_{k\in\mathcal{K}}G_k(P_k,\eta)-K\right)^2}+\frac{D\sigma_n^2}{\eta}$ is a unary quadratic function about $\frac{1}{\sqrt{\eta}}$, thus it is easy to derive the optimal solution (\ref{equ:zero_denoising_factor}) in closed form.
Corollary \ref{corollary:beta_zero} is thus proven.
\end{proof}

Note that $l^*=K$ when $\beta\rightarrow 0$ based on the proof.
The direction of the gradient vector received from each device at the edge server is independent to the power of the transmitting device.
Thus, increasing the power of all devices can reduce the noise-induced error when the composite signal misalignment error is fixed.
\begin{remark}
The threshold-based power control in \cite{cao2019optimal} and the full power transmission are sub-optimal in general and they are optimal only when $\beta\rightarrow\infty$ and $\beta\rightarrow 0$, respectively.
In particular, the threshold-based power control can only perform efficiently when the model training converges and the dataset is highly non-IID, the full power transmission can only perform efficiently when the model training just begins and the dataset is IID.
Compared with these two schemes, our proposed power control policy is always optimal.
These findings shall be validated via experimental results in Section V.
\end{remark}

\section{Adaptive Power Control with Unknown Gradient Statistics}
In this section, we consider the practical scenario where the gradient statistics $\alpha(t)$ and $\beta(t)$ are unknown.
We propose a method to estimate $\alpha(t)$ and $\beta(t)$ in each time block and then devise an adaptive power control scheme based on the optimal solution of problem $\mathcal{P}_1$ by using the estimated $\alpha(t)$ and $\beta(t)$.

\subsection{Parameters Estimation}
In this subsection, we propose a method to estimate $\alpha(t)$ and $\beta(t)$ at each time block $t$ directly based on their definitions in (\ref{equ:alpha}) and (\ref{equ:beta}), respectively.

\subsubsection{Estimation of $\alpha(t)$}
Let $B_{k}(t)\triangleq\|\bm{g}_k(t)\|$ denote the gradient norm of device $k$ at iteration $t$.
At the end of each round of local training, we let each device transmit its gradient norm $B_k(t)$ to the edge server before sending the gradient vector $\bm{g}_k(t)$.
Note that the communication cost of transmitting $B_k(t)$ is negligible compared with that of transmitting the gradient $\bm{g}_k(t)$.
This is because the gradient norm $B_k(t)$ is a scalar while the gradient $\bm{g}_k(t)$ is a vector whose dimension $D$ can be very large\footnote{In the considered datasets for experimental validation, $D$ is in the order of $10^4$.}.
By definition (\ref{equ:alpha}), we can estimate the gradient MSN as the mean of the squared gradient norms across all the participating devices:
\begin{equation}
\hat\alpha(t)=\frac{1}{K}\sum_{k\in\mathcal{K}}{B_k^2(t)}.
\label{equ:estimate_alpha}
\end{equation}

\subsubsection{Estimation of $\beta(t)$}
By definition in (\ref{equ:beta}), the gradient SMCV $\beta(t)$ depends on $m_d(t)$ and $\sigma_d(t)$.
Before each device sending its gradient at time block $t$, we cannot estimate $\beta(t)$ in advance.
However, from the experimental results of real datasets in Fig.~\ref{fig:alpha_beta}, it can be observed that $\beta(t)$ is highly correlated among adjacent iterations.
Thus we propose to estimate $\beta(t)$ using the aggregated gradient at time block $t-1$ as below
\begin{equation}
\hat\beta(t)=\frac{\hat\alpha(t-1)-\sum_{d=1}^D{\hat{g}_d^2(t-1)}}{\sum_{d=1}^D{\hat{g}_d^2(t-1)}},
\label{equ:estimate_beta}
\end{equation}
where $\hat\alpha(t-1)$ estimates $\sum_{d}{\sigma_d^2(t-1)+m_d^2(t-1)}$ and $\sum_d{\hat{g}_d^2(t-1)}$ estimates $\sum_{d}{m_d^2(t-1)}$.

\subsection{FL with Adaptive Power Control}
\begin{algorithm}[t]
\caption{FL Process with Adaptive Power Control}
\label{al:scheme}
\small
\begin{algorithmic}[1]
\State Initialize $\bm{w}(0)$ in edge server, $\hat\beta(1)$;
\For{time block $t=1,...,T$}
\State Edge server broadcasts the global model $\bm{w}(t)$ to all edge devices $k\in\mathcal{K}$;
\For{each device $k\in\mathcal{K}$ in parallel}
\State $\bm{g}_k(t)=\nabla L_{k,t}^{SGD}(\bm{w}(t))$;
\State $B_k(t)=\sqrt{\sum_d{g_{k,d}^2(t)}}$;
\State Upload $B_k(t)$ to edge server;
\EndFor
\State Edge server estimates $\hat\alpha(t)$ based on (\ref{equ:estimate_alpha});
\State Edge server obtains the optimal transmit power $\bm{p}^*(t)$ based on (\ref{equ:power_control}) and the optimal denoising factor $\eta^*(t)$ based on (\ref{equ:eta_function});
\State Edge server sends $p_k^*(t)$ to device $k$ for all $k\in\mathcal{K}$;
\For{each device $k\in\mathcal{K}$ in parallel}
\State Transmit gradient $\bm{g}_k(t)$ with power $p_k^*(t)$ to edge server using AirComp;
\EndFor
\State Edge server receives $\bm{y}(t)$ and recovers $\bm{\hat{g}}(t)$ based on (\ref{recover_gradient});
\State Edge server estimates $\hat\beta(t+1)$ based on (\ref{equ:estimate_beta});
\State Edge server updates global model $\bm{w}(t+1)=\bm{w}(t)-\gamma\bm{\hat{g}}(t)$;
\EndFor
\State Edge server returns $\bm{w}(T+1)$;
\end{algorithmic}
\end{algorithm}

In this subsection, we propose the FL process with adaptive power control, which is presented in Algorithm \ref{al:scheme}.
The algorithm has three steps.
First, each device locally takes one step of SGD on the current model using its local dataset (line 5).
After that each device calculates the norm of its local gradient and uploads it to the edge server with conventional digital transmission (line 6 and line 7).
Second, the edge server estimates parameters $\alpha(t)$ and $\beta(t)$ based on the received gradient norm at time block $t$ and historical aggregated gradient (line 9 and line 16).
Then the optimal transmit power and denoising factor are obtained based on (\ref{equ:power_control}) and (\ref{equ:eta_function}), respectively (line 10).
Third, the edge server informs the optimal transmit power to each device and each device transmits local gradient with the assigned power simultaneously using AirComp to the edge server in an analog manner (line 12-14).

Note that the computational complexity of finding the optimal power control in Algorithm \ref{al:scheme} mainly consists of the time complexity of sorting devices by aggregation capability, which is $O(K\log{K})$, and the time complexity of finding $l^*$ by Lemma \ref{lemma:condition_2_optimal}, which is $O(K)$.
It is approximately the same as the existing scheme \cite{cao2019optimal}.
The additional overhead introduced by our method when compared with \cite{cao2019optimal} lies at the communication cost for each device to send the gradient norm $B_k(t)$ to the edge server as shown in line 7 of Algorithm \ref{al:scheme}, which is negligible as discussed in Section IV-A.

\section{Experimental Results}
In this section, we provide experimental results to validate the performance of the proposed power control for AirComp-based FL over fading channels.
\subsection{Experiment Setup}
We conducted experiments on a simulated environment where the number of edge devices is $K=10$ if not specified otherwise.
The wireless channels from each device to the edge server follow IID Rayleigh fading, such that $h_k$'s are modeled as IID complex Gaussian variables with zero mean and unit variance.
For each device $k\in\mathcal{K}$, we define $\mbox{SNR}_k=\mathbb{E}\left[\frac{P_k|h_k|^2}{D\sigma_n^2}\right]=\frac{P_k}{D\sigma_n^2}$ as the average received SNR.

\subsubsection{Baselines}
We compare the proposed adaptive power control scheme with the following baseline approaches:
\begin{itemize}
\item\emph{Error-free transmission:} the aggregated gradient is updated perfectly without any transmission error.
This is equivalent to the centralized SGD algorithm.
\item\emph{Power control with known statistics:} We assume that the gradient statistics are known in advance at the beginning of the training, and then apply the proposed power control. In our experiments, the actual gradient statistics are obtained from $1000$ gradient samples without transmission error.
\item\emph{Threshold-based power control in \cite{cao2019optimal}:} this is the power control scheme given in \cite{cao2019optimal}, which assumed that signals are normalized.
    Note that it is actually the special case of our proposed power control scheme with $\beta\rightarrow\infty$ by considering the individual misalignment error only in problem $\mathcal{P}_1$.
\item\emph{Full power transmission:} all devices transmit with full power $P_k$ and the edge server applies the optimal denoising factor in (\ref{equ:eta_function}), where $l^*=K$.
\end{itemize}
\subsubsection{Datasets}
We evaluate the training of convolutional neural network on three datasets: MNIST, CIFAR-10 and SVHN.
MNIST dataset consists of 10 categories ranging from digit ¡°0¡± to ¡°9¡± and a total of 70000 labeled data samples (60000 for training and 10000 for testing).
CIFAR-10 dataset includes 60000 color images (50000 for training and 10000 for testing) of 10 different types of objects.
SVHN is a real-world image dataset for developing machine learning and object recognition algorithms with minimal requirement on data preprocessing and formatting, which includes 99289 labeled data samples (73257 for training and 26032 for testing).
\subsubsection{Data Distribution}
To study the impact of the SMCV of gradient $\beta$ for optimal transmit power, we simulate two types of dataset partitions among the mobile devices, i.e., the IID setting and non-IID one.
For the former, we randomly partition the training samples into $100$ equal shards, each of which is assigned to one particular device.
While for the latter, we first sort the data by digit label, divide it into $200$ equal shards, and randomly assign $2$ shards to each device.
\subsubsection{Neural network and hyper-parameters}
In all our experiments, we adopt the convolutional neural network (the number of layers is 11).
The hyper-parameters are set as follows: momentum optimizer is 0.5, the number of local updates between two global aggregations is 1, local batch size is 10 and learning rate $\gamma=0.01$.

\subsection{Experimental Results}
\begin{figure}[t]
\centering

\subfigure[MNIST dataset with IID partition.]
{\begin{minipage}[t]{0.45\linewidth}
\centering
\includegraphics[width=2.6in]{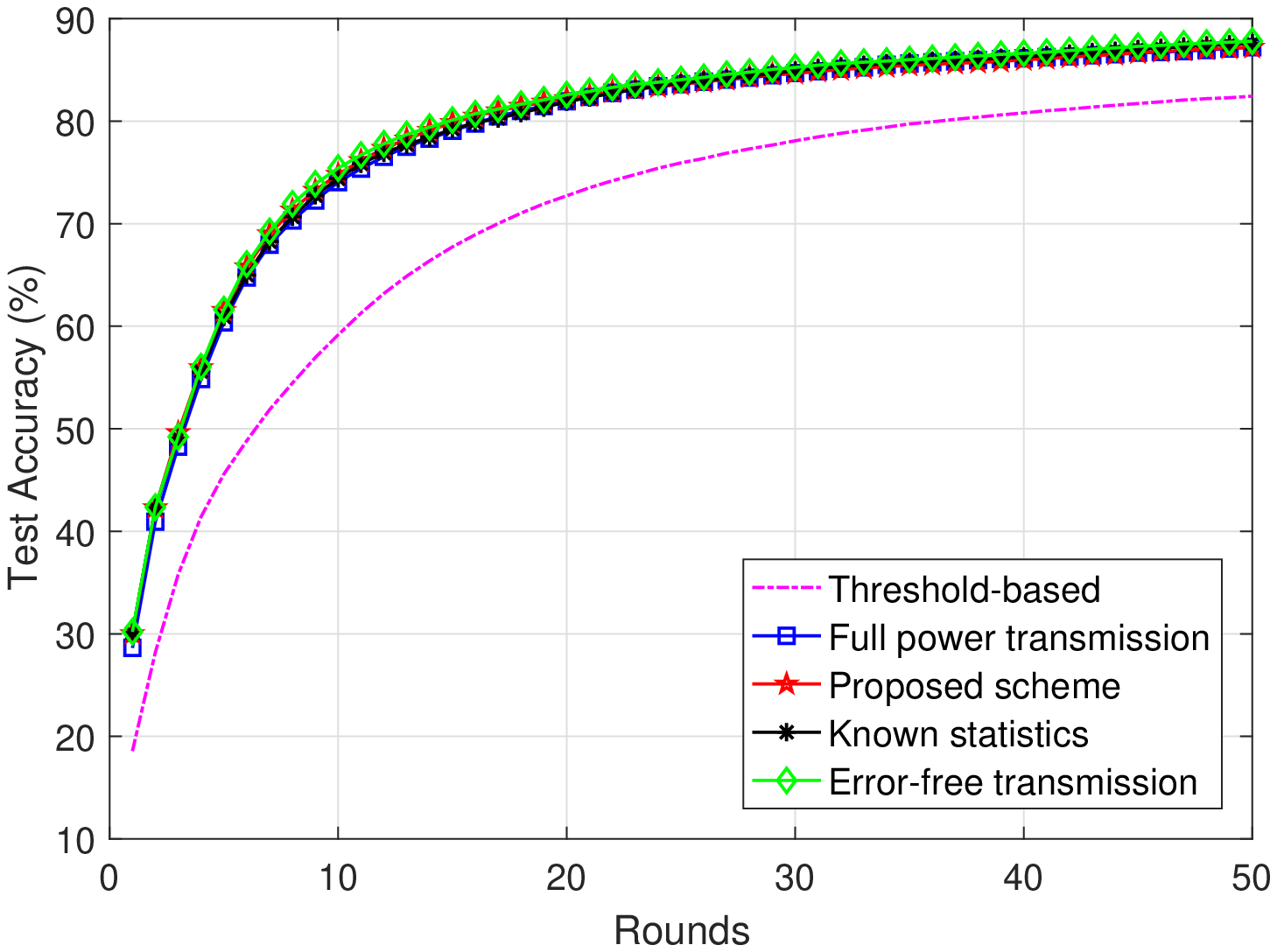}
\label{fig:mnist_iid}
\end{minipage}}
\subfigure[MNIST dataset with non-IID partition.]
{\begin{minipage}[t]{0.45\linewidth}
\centering
\includegraphics[width=2.6in]{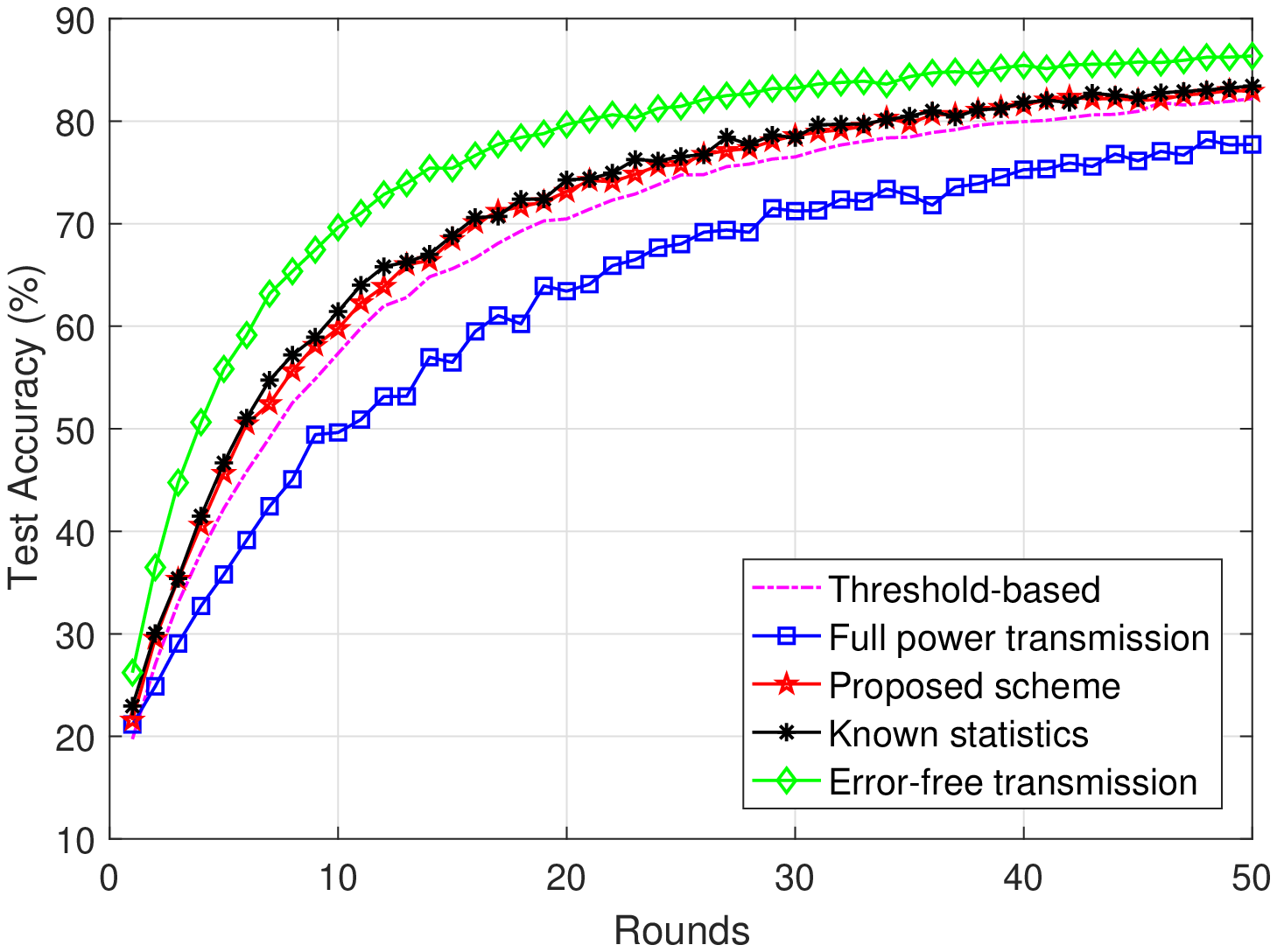}
\label{fig:mnist_non-iid}
\end{minipage}}

\subfigure[CIFAR-10 dataset with IID partition.]
{\begin{minipage}[t]{0.45\linewidth}
\centering
\includegraphics[width=2.6in]{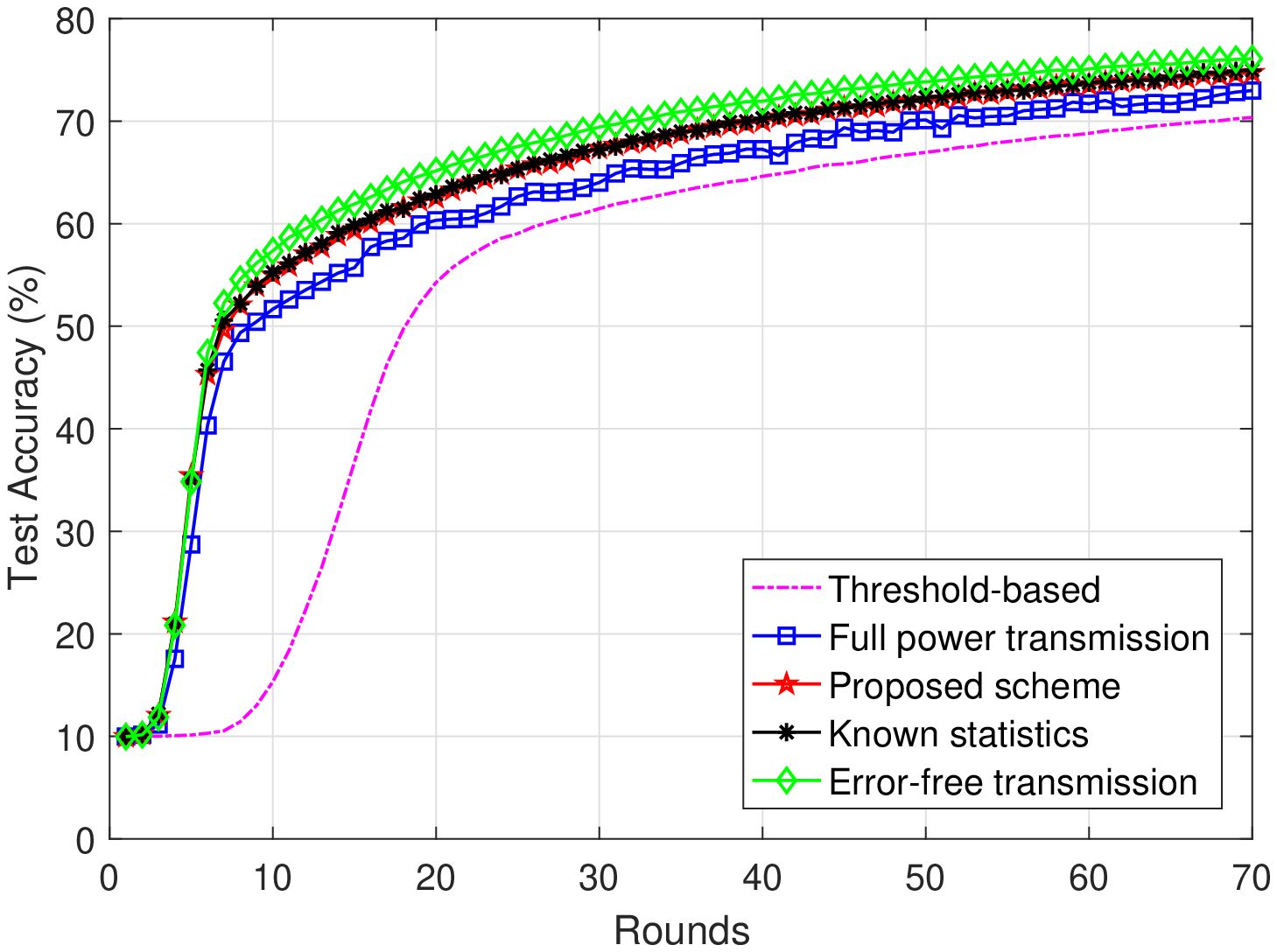}
\label{fig:cifar_iid}
\end{minipage}}
\subfigure[CIFAR-10 dataset with non-IID partition.]
{\begin{minipage}[t]{0.45\linewidth}
\centering
\includegraphics[width=2.6in]{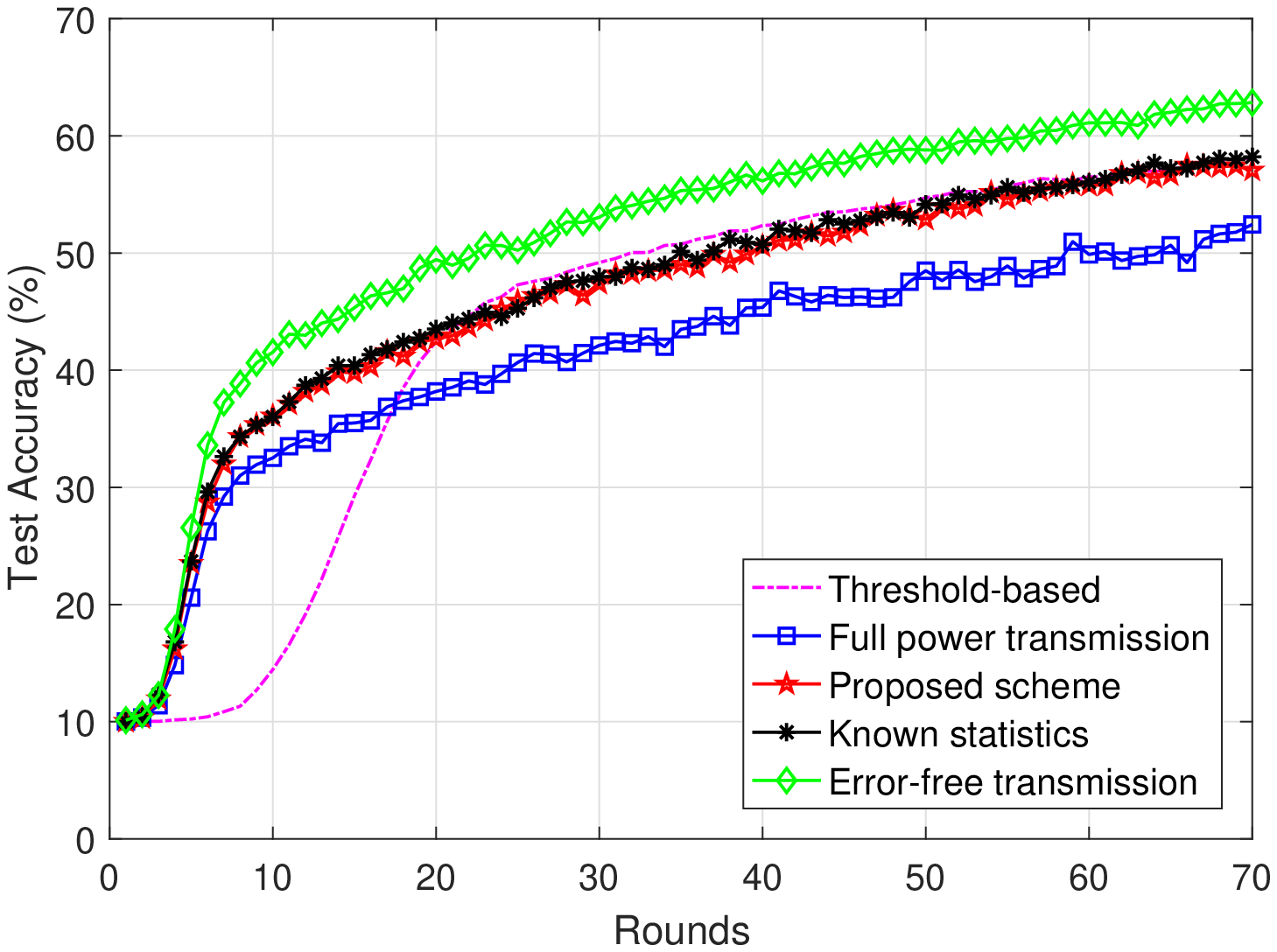}
\label{fig:cifar_non-iid}
\end{minipage}}

\subfigure[SVHN dataset with IID partition.]
{\begin{minipage}[t]{0.45\linewidth}
\centering
\includegraphics[width=2.6in]{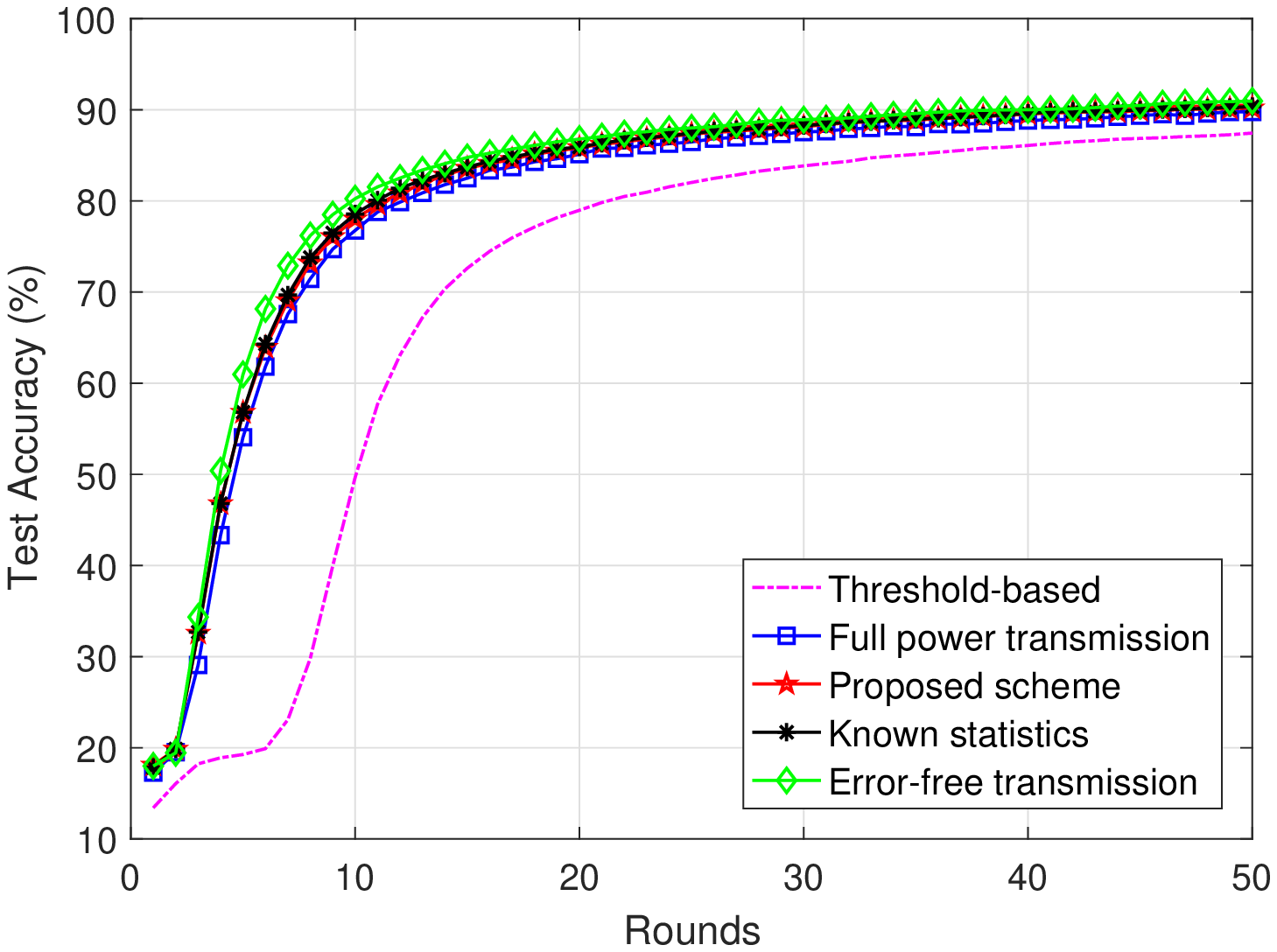}
\label{fig:SVHN_iid}
\end{minipage}}
\subfigure[SVHN dataset with non-IID partition.]
{\begin{minipage}[t]{0.45\linewidth}
\centering
\includegraphics[width=2.6in]{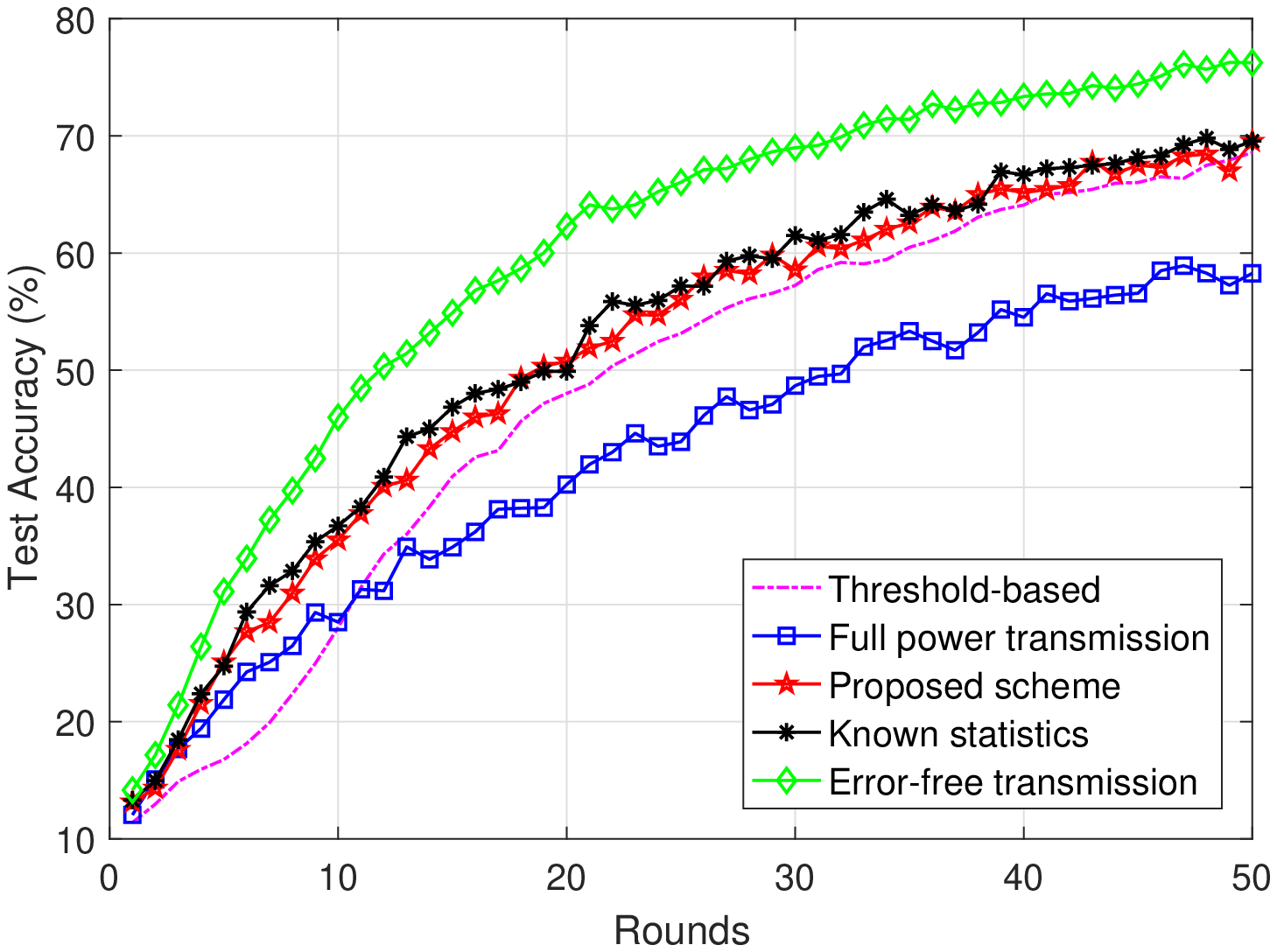}
\label{fig:SVHN_non-iid}
\end{minipage}}

\centering
\caption{Performance comparison on different dataset partition, for $K=10$ and $\mbox{SNR}_k=10\mbox{ dB}, \forall k\in\mathcal{K}$.}
\label{fig:varying beta} 
\end{figure}
Fig.~\ref{fig:varying beta} compares the test accuracy for the three considered datasets with IID dataset partition and non-IID dataset partition, respectively, where the average received SNR is set as 10 dB for all devices.
It is observed that the performance gap to the scheme with known gradient statistics is very small, which indicates the proposed methods for estimating gradient statistics are effective.
It is also observed that the model accuracy of the proposed power control scheme is better than threshold-based power control and full power transmission.
From Fig.~\ref{fig:alpha_beta}, we have known that the averaged gradient SMCV $\beta(t)$ in the IID dataset partition is less than that in the non-IID dataset partition and it increases over iterations.
Threshold-based power control suffers significant accuracy loss in the IID partition or at the beginning of training.
This is because in this case, the gradient SMCV is small and thus the MSE is dominated by the composite misalignment error.
As a result, threshold-based power control that considers the individual misalignment error only is much inferior.
Besides, full power transmission suffers significant accuracy loss in the non-IID partition or at the end of training.
This is because the gradient SMCV is large and therefore the full power transmission scheme fails to minimize the individual misalignment error that dominates the MSE in this case.

\begin{figure}[t]
\begin{centering}
\vspace{-0.2cm}
\includegraphics[scale=.50]{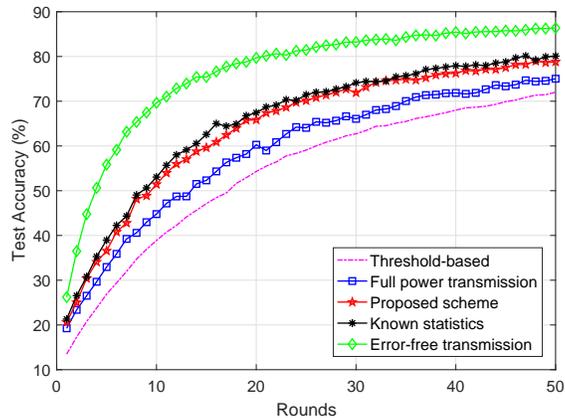}
\vspace{-0.1cm}
 \caption{\small{Performance comparison for MNIST dataset with non-IID partition at the average received $\mbox{SNR}=5\mbox{ dB}$.}}\label{fig:varying SNR}
\end{centering}
\vspace{-0.3cm}
\end{figure}
Fig.~\ref{fig:varying SNR} illustrates the test accuracy for MNIST with non-IID data partition at the average received $\mbox{SNR}=\mbox{5 dB}$.
It is observed that the overall performance of the proposed scheme is still better than two baseline approaches at low SNR region.
In specific, full power transmission performs better than threshold-based power control scheme.
This is mainly because when the noise variance is large, full power transmission can strongly suppress noise error that dominates the MSE.

\begin{figure}[t]
\begin{centering}
\vspace{-0.2cm}
\includegraphics[scale=.50]{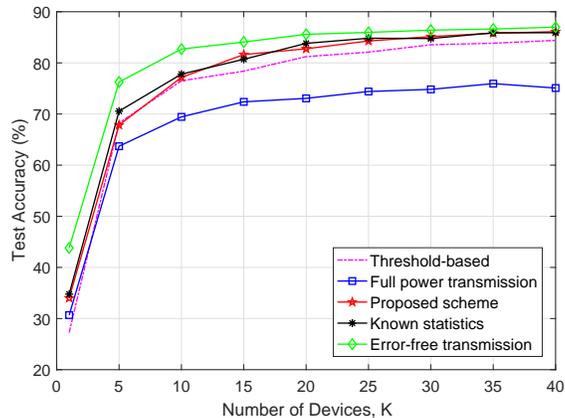}
\vspace{-0.1cm}
 \caption{\small{Performance comparison over the number of devices, where MNIST dataset is non-IID partition and $\mbox{SNR}_k=10\mbox{ dB}, \forall k\in\mathcal{ K}$.}}\label{fig:minst_devices}
\end{centering}
\vspace{-0.3cm}
\end{figure}
Finally, Fig.~\ref{fig:minst_devices} compares the test accuracy of different power control schemes at varying number of devices $K$.
Here, MNIST dataset with non-IID partition is used, the average received SNR of all the devices is set as $\mbox{SNR}_k=10$ dB and the results are averaged over $50$ model trainings.
First, it is observed that the test accuracy achieved by all the four schemes increases when the number of participating devices $K$ increases but cannot increase further when $K$ is large enough.
In particular, the test accuracy of all the considered schemes keeps unchanged when $K\geq 30$, due to the fact that the edge server can aggregate enough data for averaging.
Second, the proposed scheme outperforms both of threshold-based power control and full power transmission throughout the whole regime of $K$.
Full power transmission approaches threshold-based power control when $K$ is small (i.e., $K = 4$ in Fig.~\ref{fig:minst_devices}), but the performance compromises as $K$ increases, due to the lack of power adaptation to reduce the misalignment error.
\section{Conclusion}
This work studied the power control optimization problem for the over-the-air federated learning over fading channels by
taking the gradient statistics into account.
The optimal power control policy is derived in closed form when the first- and second-order gradient statistics are known.
It is shown that the optimal transmit power on each device decreases with gradient SMCV and increases with noise variance.
In the special cases where $\beta$ approaches infinity and zero, the optimal transmit power reduces to threshold-based power control and full power transmission, respectively.
We propose an adaptive power control algorithm that dynamically adjusts the transmit power in each iteration based on the estimation results.
Experimental results show that our proposed adaptive power control scheme outperforms the existing schemes.
In the future work, we can further exploit the gradient statistics for joint power control and beamforming when each edge device has multiple antennas.

\begin{appendices}
\section{Proof of Lemma \ref{lemma:eta_lower_bound}}
\label{proof:lemma:eta_lower_bound}
We prove this Lemma by contradiction.
Suppose the optimal denoising factor $\eta^*\leq C_1^2$.
Both the individual and composite signal misalignment errors can be forced to zero by letting $p_k^*=\frac{\eta\alpha}{|h_k|^2},\forall k\in\mathcal{K}$.
The problem $\mathcal{P}_1$ can thus be expressed as
\begin{equation}
\min_{\eta\leq C_1^2}{\frac{D\sigma_n^2}{\eta}}.
\end{equation}
It is obvious that the optimal solution of the above problem is $\eta^*=C_1^2$.
Thus, it must hold that $\eta\geq C_1^2$ for problem $\mathcal{P}_1$.

\section{Proof of Lemma \ref{lemma:optimal_power_constrain}} \label{proof:lemma:optimal_power_constrain}
To prove this lemma, we need to prove that if $p_{k_2}^* = P_{k_2}$ for some device $k_2$, we shall have $p_{k_1}^* = P_{k_1}, \forall k_1 < k_2$.
We prove this by contradiction.
Let $\bm{p}^*=[p_1^*,...,p_K^*]$ denote the optimal transmit power to the problem $\mathcal{P}_1$.
We assume that there are two devices $k_1<k_2$ satisfying $p_{k_1}^*<P_{k_1}$ and $p_{k_2}^*=P_{k_2}$.
Then there always exists a modified transmit power $\bm{p}^{'}=[p_1^*,...,p_{k_1-1}^*,p_{k_1}^{'},p_{k_1+1}^*,...,p_{k_2-1}^*,p_{k_2}^{'},p_{k_2+1}^*,...,p_K^*]$, where $p_{k_1}^{'}=P_{k_1}$ and $p_{k_2}^{'}<P_{k_2}$, satisfying $\frac{\sqrt{p_{k_1}^*}|h_{k_1}|}{\sqrt{\eta^*\alpha}}+\frac{\sqrt{p_{k_2}^*}|h_{k_2}|}{\sqrt{\eta^*\alpha}}=\frac{\sqrt{p_{k_1}^{'}}|h_{k_1}|}{\sqrt{\eta^*\alpha}}+\frac{\sqrt{p_{k_2}^{'}}|h_{k_2}|}{\sqrt{\eta^*\alpha}}$.
The resulting MSE obtained by $\bm{p}^{'}$ only differs from the minimum MSE by $\bm{p}^*$ in the individual misalignment error term.
The difference is given by
\begin{align}
&\mbox{MSE}(\bm{p}^*)-\mbox{MSE}(\bm{p}^{'})\nonumber\\
=&\frac{\beta\alpha}{\beta+1}\left[\left(\frac{\sqrt{p_{k_1}^*}|h_{k_1}|}{\sqrt{\eta^*\alpha}}\right)^2+\left(\frac{\sqrt{p_{k_2}^*}|h_{k_2}|}{\sqrt{\eta^*\alpha}}\right)^2-\left(\frac{\sqrt{p_{k_1}^{'}}|h_{k_1}|}{\sqrt{\eta^*\alpha}}\right)^2-\left(\frac{\sqrt{p_{k_2}^{'}}|h_{k_2}|}{\sqrt{\eta^*\alpha}}\right)^2\right]\nonumber\\
=&\frac{2\beta\alpha}{\beta+1}\left(\frac{\sqrt{p_{k_2}^*}|h_{k_2}|}{\sqrt{\eta^*\alpha}}-\frac{\sqrt{p_{k_2}^{'}}|h_{k_2}|}{\sqrt{\eta^*\alpha}}\right)\left(\frac{\sqrt{p_{k_2}^*}|h_{k_2}|}{\sqrt{\eta^*\alpha}}-\frac{\sqrt{p_{k_1}^{'}}|h_{k_1}|}{\sqrt{\eta^*\alpha}}\right)\nonumber\\
\geq&0.\label{equ:proof:lemma:optimal_power_constrain}
\end{align}
The inequality in (\ref{equ:proof:lemma:optimal_power_constrain}) holds as $p_{k_2}^*=P_{k_2}>p_{k_2}^{'}$ and $\frac{\sqrt{p_{k_2}^*}|h_{k_2}|}{\sqrt{\eta^*\alpha}}\geq\frac{\sqrt{p_{k_1}^{'}}|h_{k_1}|}{\sqrt{\eta^*\alpha}}$ by the aggregation capability ranking in (\ref{equ:sort}).
This indicates that $\bm{p}^{'}$ is a better solution than $\bm{p}^*$, which contradicts the assumption.
Therefore, for all pairs of two devices $k_1<k_2$, if $p_{k_2}^*=P_{k_2}$, we must have $p_{k_1}^*=P_{k_1}$.
Lemma \ref{lemma:optimal_power_constrain} is proved.

\section{Proof of Lemma \ref{lemma:illegal_solution}} \label{proof:lemma:illegal_solution}
To prove this lemma is equal to proving that if the optimal solution $\bm{\tilde{p}}^*_l$ of the sub-problem defined in the relaxed subregion $\tilde{\mathcal{M}_l}$ does not lie in $\mathcal{M}_l$, the global optimal solution $p^*$ of problem $\mathcal{P}_1$ must not lie in subregion $\mathcal{M}_l$.
First, we prove that $\bm{\tilde{p}}^*_l$ is the only local minimum point in $\tilde{\mathcal{M}_l}$.
Note that the derivative of a differentiable function at the local minimum point must be 0, and $\bm{\tilde{p}}^*_l$ is the only zero point of the derivative of (\ref{equ:problem_formulation}) w.r.t. $p_k$.
Hence, $\bm{\tilde{p}}^*_l$ is the only local minimum point in $\tilde{\mathcal{M}_l}$.
Now, we prove Lemma \ref{lemma:illegal_solution} by contradiction.
We assume that $\bm{\tilde{p}}^*_l\notin\mathcal{M}_l$, then there is no local minimum point in $\mathcal{M}_l$ due to $\mathcal{M}_l\subseteq\tilde{\mathcal{M}_l}$.
If $\bm{p}^*$ lies in $\mathcal{M}_l$, $\bm{p}^*$ must be a local minimum point due to the fact that $\bm{p}^*$ is the minimum point in $\mathcal{M}_l$ and $\bm{p}^*$ is an interior point of $\mathcal{M}_l$.
It contradicts that there is no local minimum point in $\mathcal{M}_l$.
Therefore, if $\bm{\tilde{p}}^*_l$ does not lie in $\mathcal{M}_l$, the global optimal transmit power $\bm{p}^*$ must not lie in the $l$-th subregion $\mathcal{M}_l$.
The proof of Lemma \ref{lemma:illegal_solution} is completed.

\section{Proof of Theorem \ref{theorem:optimality}} \label{proof:theorem:optimality}
To complete the proof, we need to show that with $l^*$ defined in (\ref{equ:optimal_tau}), the optimal transmit power is $\bm{\tilde{p}}^*_{l^*}$ in the $l^*$-th relaxed subregion $\tilde{\mathcal{M}_{l^*}}$.
Lemma \ref{lemma:illegal_solution} shows that $\forall l\in\mathcal{K}$, if $\bm{\tilde{p}}^*_l\notin\mathcal{M}_l$, the optimal transmit power $\bm{p}^*$ must not lie in $\mathcal{M}_l$.
By definition of $\mathcal{L}$, $\forall l\in\mathcal{K}\setminus\mathcal{L}$, $\bm{\tilde{p}}^*_l\notin\mathcal{M}_l$.
Hence, $\bm{p}^*$ must lie in $\bigcup_{l\in\mathcal{L}}{\mathcal{M}_l}$, i.e.,
\begin{align}
\mbox{MSE}(\bm{p}^*)=\min_{\bm{p}\in\bigcup_{l\in\mathcal{L}}{\mathcal{M}_l}}{\mbox{MSE}(\bm{p})}.
\label{equ:optimal_region}
\end{align}
For all $l\in\mathcal{L}$, $\bm{\tilde{p}}^*_l$ is the minimum point in $\tilde{\mathcal{M}_l}$, and it is also the minimum point in $\mathcal{M}_l$ due to $\bm{\tilde{p}}^*_l\in\mathcal{M}_l$ and $\mathcal{M}_l\subseteq\tilde{\mathcal{M}_l}$, i.e.,
\begin{align}
\min_{\bm{p}\in\mathcal{M}_l}{\mbox{MSE}(\bm{p})}=\mbox{MSE}(\bm{\tilde{p}}^*_l)=V_l,\forall l\in\mathcal{L}.
\label{equ:optimal_subregion}
\end{align}
Substituting (\ref{equ:optimal_subregion}) back to (\ref{equ:optimal_region}), we have
\begin{align}
\mbox{MSE}(\bm{p}^*)=\min_{\bm{p}\in\bigcup_{l\in\mathcal{L}}{\mathcal{M}_l}}{\mbox{MSE}(\bm{p})}=\min_{l\in\mathcal{L}}{V_l}=V_{l^*}.
\end{align}

Therefore, the candidate transmit power $\bm{\tilde{p}}^*_{l^*}$ with the smallest value $V_{l^*}$ is the optimal transmit power of the problem $\mathcal{P}_1$.
According to (\ref{equ:power_control_subregion}) and (\ref{equ:eta_function_subregion}), the optimal transmit power and denoising factor are the forms given in Theorem \ref{theorem:optimality}.

\section{Proof of Lemma \ref{lemma:condition_2_optimal}} \label{proof:lemma:condition_2_optimal}
When $\bm{p}^*$ lies in subregion $\mathcal{M}_l$, $\bm{p}^*$ is equal to $\bm{\tilde{p}}^*_l$.
Thus, to prove the sufficiency of this Lemma is equal to prove that the optimal transmit power $\bm{p}^*$ holds
\begin{align}
C_{l^*}\leq\frac{\sqrt{p_k^*}|h_k|}{\sqrt{\alpha}}<C_{l^*+1}, \forall{k}\in\{l^*+1,...,K\}.
\label{equ:first_inequality}
\end{align}
Based on (\ref{equ:power_control}), $p_k^*<P_k$, for $k\in\{l^*+1,...,K\}$, and when $k\in\{l^*+1\}$, we have $\frac{\sqrt{p_{l^*+1}^*}|h_{l^*+1}|}{\sqrt{\alpha}}<C_{l^*+1}$.
Since $\frac{\sqrt{p_k^*}|h_k|}{\sqrt{\alpha}}, \forall{k}\in\{l^*+1,...,K\}$ is the same, we have the inequality $\frac{\sqrt{p_k^*}|h_k|}{\sqrt{\alpha}}<C_{l^*+1}$, for $k\in\{l^*+1,...,K\}$.
We prove $C_{l^*}\leq\frac{\sqrt{p_k^*}|h_k|}{\sqrt{\alpha}}, \forall{k}\in\{l^*+1,...,K\}$ by contradiction.
We assume that $C_{l^*}>\frac{\sqrt{p_{l^*+1}^*}|h_{l^*+1}|}{\sqrt{\alpha}}$.
As $p_k^*=P_k$, for $k\in\{1,...,l^*\}$, we have $\frac{\sqrt{p_{l^*}^*}|h_{l^*}|}{\sqrt{\alpha}}=C_{l^*}>\frac{\sqrt{p_{l^*+1}^*}|h_{l^*+1}|}{\sqrt{\alpha}}$.
Then there always exists a modified transmit power $\bm{p}^{'}=[p_1^*,...,p_{l^*-1}^*,p_{l^*}^{'},p_{l^*+1}^{'},p_{l^*+2}^*,...,p_K^*]$ where transmit power $p_{l^*}^{'}$ and $p_{l^*+1}^{'}$ satisfy $\frac{\sqrt{p_{l^*}^{'}}|h_{l^*}|}{\sqrt{\alpha}}=\frac{\sqrt{p_{l^*+1}^{'}}|h_{l^*+1}|}{\sqrt{\alpha}}=\frac{1}{2}(\frac{\sqrt{p_{l^*}^*}|h_{l^*}|}{\sqrt{\alpha}}+\frac{\sqrt{p_{l^*+1}^*}|h_{l^*+1}|}{\sqrt{\alpha}})$.
The difference between MSE of the transmit power $\bm{p}^*$ and $\bm{p}^{'}$ is given by
\begin{align}
&\mbox{MSE}(\bm{p}^*)-\mbox{MSE}(\bm{p}^{'})\nonumber\\
=&\frac{\beta\alpha}{\beta+1}\left[\left(\frac{\sqrt{p_{l^*}^*}|h_{l^*}|}{\sqrt{\eta^*\alpha}}\right)^2+\left(\frac{\sqrt{p_{l^*+1}^*}|h_{l^*+1}|}{\sqrt{\eta^*\alpha}}\right)^2-\left(\frac{\sqrt{p_{l^*}^{'}}|h_{l^*}|}{\sqrt{\eta^*\alpha}}\right)^2-\left(\frac{\sqrt{p_{l^*+1}^{'}}|h_{l^*+1}|}{\sqrt{\eta^*\alpha}}\right)^2\right]\nonumber\\
=&\frac{\beta\alpha}{2\left(\beta+1\right)}\bigg[\left(\frac{\sqrt{p_{l^*}^*}|h_{l^*}|}{\sqrt{\eta^*\alpha}}-\frac{\sqrt{p_{l^*+1}^*}|h_{l^*+1}|}{\sqrt{\eta^*\alpha}}\right)^2\bigg]>0.
\label{equ:averaged_power_is_better}
\end{align}
This indicates that $\bm{p}^{'}$ is a better solution than $\bm{p}^*$, which contradicts the assumption.
We have the inequality $C_{l^*}\leq\frac{\sqrt{p_k^*}|h_k|}{\sqrt{\alpha}}$, for $k\in\{l^*+1,...,K\}$.
Thus, the sufficiency of this Lemma has been proved.

To prove the necessity of this Lemma, we prove the inverse negative proposition of it: if the optimal transmit power $\bm{p}^*$ does not lie in subregion $\mathcal{M}_l$, i.e., $l\in\{1,...,l^*-1,l^*+1,...,K\}$, $\bm{\tilde{p}}^*_l$ satisfies: $\frac{\sqrt{\tilde{p}^*_{l,k}}|h_k|}{\sqrt{\alpha}}\geq C_{l+1}$ or $\frac{\sqrt{\tilde{p}^*_{l,k}}|h_k|}{\sqrt{\alpha}}<C_{l}$ for all $k\in\{l+1,...,K\}$.

First, we prove that $\bm{\tilde{p}^*_l}$, for $\forall l\in\{1,...,l^*-1\}$ holds
\begin{align}
\frac{\sqrt{\tilde{p}^*_{l,k}}|h_{k}|}{\sqrt{\alpha}}\geq C_{l+1}, \forall k\in\{l+1,...,K\}.
\label{equ:second_inequality}
\end{align}
$\forall l\in\{1,...,l^*-1\}$, $\mathcal{M}_{l^*}\subseteq\tilde{\mathcal{M}_l}$, then we have $\mbox{MSE}(\bm{\tilde{p}^*_l})\leq\mbox{MSE}(\bm{p}^*)$.
We prove (\ref{equ:second_inequality}) by contradiction.
If $\exists{l}\in\{1,...,l^*-1\},\frac{\sqrt{\tilde{p}^*_{l,k}}|h_{k}|}{\sqrt{\alpha}}<C_{l+1},\forall k\in\{l+1,...,K\}$, i.e., $\bm{\tilde{p}^*_l}\in\mathcal{M}_l$, the feasible transmit power $\bm{\tilde{p}}^*_l$ is better than the optimal transmit power $\bm{p}^*$, which contradicts the assumption.
We have the inequality (\ref{equ:second_inequality}).

Second, we prove that $\bm{\tilde{p}^*_l}$, for $\forall l\in\{l^*+1,...,K\}$ holds
\begin{align}
\frac{\sqrt{\tilde{p}^*_{l,k}}|h_{k}|}{\sqrt{\alpha}}<C_l, \forall k\in\{l+1,...,K\}.
\label{equ:third_inequality}
\end{align}
We prove this by contradiction.
When $l=l^*+1$, we assume that $\frac{\sqrt{\tilde{p}^*_{l,k}}|h_{k}|}{\sqrt{\alpha}}\geq C_l,\forall k\in\{l+1,...,K\}$.
Let $\bm{p}^{avg}=[\tilde{p}^*_{l,1},...,\tilde{p}^*_{l,l-1},p_{l}^{avg},...,p_K^{avg}]$ denote a modified $\bm{\tilde{p}}^*_l$, where $\frac{\sqrt{p_{l}^{avg}}|h_{l}|}{\sqrt{\alpha}}=...=\frac{\sqrt{p_{K}^{avg}}|h_{K}|}{\sqrt{\alpha}}=\frac{1}{K-l+1}\sum_{i=l}^{K}{\frac{\sqrt{\tilde{p}^*_{l,i}}|h_{i}|}{\sqrt{\alpha}}}$.
Using a proof similar to that of (\ref{equ:averaged_power_is_better}), we can prove that $\mbox{MSE}(\bm{p}^{avg})\leq\mbox{MSE}(\bm{\tilde{p}}^*_l)$.
Let $\bm{p}^{fes}=[\tilde{p}^*_{l,1},...,\tilde{p}^*_{l,l-1},p_{l}^{fes},...,p_K^{fes}]$ denote a modified $\bm{\tilde{p}}^*_l$, where $\frac{\sqrt{p^{fes}_k}|h_k|}{\sqrt{\alpha}}=C_l,$, for $k\in\{l,...,K\}$.
The derivative of (\ref{equ:problem_formulation}) w.r.t. $p_k$ for all $k\in\{l+1,...,K\}$ shows that the MSE increases with $p_k$ when $p_k\geq p_k^*$.
Note that $\bm{p}^{avg}$ and $\bm{p}^{fes}$ are in the $l^*$-th relaxed subregion $\mathcal{\tilde{M}}_{l^*}$ and $p^{avg}_k\geq p^{fes}_k>p_k^*$, for $k\in\{l,...,K\}$, $\mbox{MSE}(\bm{p}^{fes})\leq\mbox{MSE}(\bm{p}^{avg})$ and $\mbox{MSE}(\bm{p}^{fes})\leq\mbox{MSE}(\bm{\tilde{p}}^*_l)$.
The transmit power $\bm{p}^{fes}\in\mathcal{\tilde{M}}_l$ is better than $\bm{\tilde{p}}^*_l\in\mathcal{\tilde{M}}_l$, which contradicts that $\bm{\tilde{p}}^*_l$ is optimal solution in $\mathcal{\tilde{M}}_l$.
We have the inequality (\ref{equ:third_inequality}) when $l=l^*+1$ and it can be extended to when $l\in\{l^*+2,...,K\}$ by mathematical induction.

Thus, the necessity of this Lemma has been proved.
We complete the proof of Lemma \ref{lemma:condition_2_optimal}.

\section{Proof of Lemma \ref{lemma:continuous_globally}}
\label{proof:lemma:continuous_globally}
We first prove that the upper boundary $\mathcal{U}_l$ is equal to the lower boundary $\mathcal{L}_{l+1}$, for $l\in\{1,...,K-1\}$.
Then we prove that the optimal transmit power function $\bm{p}^*(\alpha,\beta,\sigma^2_n)$ is continuous at $\mathcal{U}_l=\mathcal{L}_{l+1}$, for $l\in\{1,...,K-1\}$.

For all $(\alpha,\beta,\sigma^2_n)\in\mathcal{L}_{l+1}$, the corresponding optimal transmit power $\bm{\tilde{p}}^*_{l+1}(\alpha,\beta,\sigma^2_n)=(\frac{C_{l+1}\sqrt{\alpha}}{|h_k|})^2$, i.e., $\frac{\sqrt{\tilde{p}_{l+1,k}^*(\alpha,\beta,\sigma^2_n)}|h_k|}{\sqrt{\alpha}}=C_{l+1}$, for $k\in\{l+2,...,K\}$.
Based on (\ref{equ:power_control_subregion}) and (\ref{equ:eta_function_subregion}), we have
\begin{subequations}
\label{equ:equal_power}
\begin{align}
& & \sqrt{\tilde{\eta}^*_{l+1}(\alpha,\beta,\sigma^2_n)}&=\frac{\sum_{i=1}^{l+1}{C_i}+(\beta+K-l-1)C_{l+1}}{\beta+K}\\
\Leftrightarrow & & \frac{\frac{\beta\alpha}{\beta+1}\sum\limits_{i=1}^{l+1}C_i^2+D\sigma_n^2}{\frac{\beta(\beta+K)\alpha}{(\beta+K-l-1)(\beta+1)}\sum\limits_{i=1}^{l+1}C_i}&=\frac{(\beta+K-l-1)C_{l+1}}{\beta+K}\\
\Leftrightarrow & & \frac{\beta\alpha}{\beta+1}\sum\limits_{i=1}^{l+1}C_i^2+D\sigma_n^2&=\frac{\beta\alpha}{\beta+1}C_{l+1}\sum_{i=1}^{l+1}{C_i}\\
\Leftrightarrow & & \frac{\beta\alpha}{\beta+1}\sum\limits_{i=1}^{l}C_i^2+D\sigma_n^2&=\frac{\beta\alpha}{\beta+1}C_{l+1}\sum_{i=1}^l{C_i}\\
\Leftrightarrow & & \frac{\frac{\beta\alpha}{\beta+1}\sum\limits_{i=1}^{l}C_i^2+D\sigma_n^2}{\frac{\beta(\beta+K)\alpha}{(\beta+K-l)(\beta+1)}\sum\limits_{i=1}^{l}C_i}&=\frac{(\beta+K-l)C_{l+1}}{\beta+K}\\
\Leftrightarrow & & \frac{\frac{\beta\alpha}{\beta+1}\sum\limits_{i=1}^{l}C_i^2+\frac{\beta\alpha}{(\beta+K-l)(\beta+1)}\bigg(\sum\limits_{i=1}^{l}C_i\bigg)^2+D\sigma_n^2}{\frac{\beta(\beta+K)\alpha}{(\beta+K-l)(\beta+1)}\sum\limits_{i=1}^{l}C_i}&=\frac{\sum_{i=1}^l{C_i}+(\beta+K-l)C_{l+1}}{\beta+K}\\
\Leftrightarrow & & \sqrt{\tilde{\eta}^*_l(\alpha,\beta,\sigma^2_n)}&=\frac{\sum_{i=1}^l{C_i}+(\beta+K-l)C_{l+1}}{\beta+K}\\
\Leftrightarrow & & \frac{\sqrt{\tilde{p}_{l,k}^*(\alpha,\beta,\sigma^2_n)}|h_k|}{\sqrt{\alpha}}&=C_{l+1}, \forall k\in\{l+1,...,K\}\\
\Leftrightarrow & & \tilde{p}^*_{l,k}(\alpha,\beta,\sigma^2_n)&=(\frac{C_{l+1}\sqrt{\alpha}}{|h_k|})^2, \forall k\in\{l+1,...,K\},
\end{align}
\end{subequations}
then we have $(\alpha,\beta,\sigma^2_n)\in\mathcal{U}_l$ by the definition of $\mathcal{U}_l$, i.e., $\mathcal{L}_{l+1}\subseteq\mathcal{U}_l$.
For all $(\alpha,\beta,\sigma^2_n)\in\mathcal{U}_l$, we have $(\alpha,\beta,\sigma^2_n)\in\mathcal{L}_{l+1}$, i.e., $\mathcal{U}_l\subseteq\mathcal{L}_{l+1}$ by reversing the derivation of (\ref{equ:equal_power}).
Therefore, we have $\mathcal{L}_{l+1}=\mathcal{U}_l$, for $l\in\{1,...,K-1\}$.

For all $(\alpha,\beta,\sigma^2_n)_0\in\mathcal{L}_{l+1}=\mathcal{U}_l$, define $(\alpha,\beta,\sigma^2_n)_+\in\mathcal{X}_{l+1}$ and $(\alpha,\beta,\sigma^2_n)_-\in\mathcal{X}_{l}$ as two points infinitely close to $(\alpha,\beta,\sigma^2_n)_0$, respectively.
Then the left limitation of $\bm{p}^*((\alpha,\beta,\sigma^2_n)_0)$ is given by
\begin{align}
\lim_{(\alpha,\beta,\sigma^2_n)_-\rightarrow(\alpha,\beta,\sigma^2_n)_0}{\bm{p}^*((\alpha,\beta,\sigma^2_n)_-)}
&=\lim_{(\alpha,\beta,\sigma^2_n)_-\rightarrow(\alpha,\beta,\sigma^2_n)_0}{\bm{\tilde{p}}^*_l((\alpha,\beta,\sigma^2_n)_-)}\nonumber\\
&=\bm{\tilde{p}}^*_l((\alpha,\beta,\sigma^2_n)_0)=\bm{\tilde{p}}^*_{l+1}((\alpha,\beta,\sigma^2_n)_0)=\bm{p}^*((\alpha,\beta,\sigma^2_n)_0).
\end{align}
Similarly, the right limitation of $\bm{p}^*((\alpha,\beta,\sigma^2_n)_0)$ is also equal to $\bm{p}^*((\alpha,\beta,\sigma^2_n)_0)$.
Therefore, the optimal transmit power function $\bm{p}^*(\alpha,\beta,\sigma^2_n)$ is continuous at $\mathcal{U}_l=\mathcal{L}_{l+1}$, for $l\in\{1,...,K-1\}$.
We complete the proof of Lemma \ref{lemma:continuous_globally}.
\end{appendices}

\bibliographystyle{IEEEtran}
\bibliography{IEEEabrv,power_control}

\end{document}